\documentclass[onefignum,onetabnum]{siamonline181217}

\usepackage{amsmath,amssymb}
\usepackage{algorithm}
\usepackage{algorithmic}
\usepackage{xcolor}
\usepackage{float}
\usepackage{setspace}
\usepackage{enumerate}
\usepackage{hyperref}
\usepackage{bm}
\usepackage{multirow}
\usepackage{cite}
\usepackage{microtype}
\usepackage{verbatim}
\usepackage{array}
\usepackage{graphicx}
\usepackage{subfigure}

\newcolumntype{C}[1]{>{\centering\let\newline\\\arraybackslash\hspace{0pt}}m{#1}}

\newsiamremark{remark}{Remark}
\newsiamremark{hypothesis}{Hypothesis}
\crefname{hypothesis}{Hypothesis}{Hypotheses}
\newsiamthm{claim}{Claim}

\headers{A unifying framework for $n$-dimensional quasi-conformal mappings}{Daoping Zhang, Gary P. T. Choi, Jianping Zhang, and Lok Ming Lui}

\title{A unifying framework for $n$-dimensional quasi-conformal mappings\thanks{Submitted to the editors DATE.
\funding{This work was supported in part by the National Science Foundation under Grant No.~DMS-2002103 (to Gary P.~T.~Choi), the National Natural Science Foundation of China (NSFC Project number 11771369), the outstanding young scholars of Education Bureau of Hunan Province, P. R. China (number 17B257) and the Natural Science Foundation of Hunan Province, P. R. China (number 2018JJ2375, 2017SK2014, 2018XK2304) (to Jianping Zhang), and HKRGC GRF under project ID 2130656 (to Lok Ming Lui).}}}

\author{Daoping Zhang\thanks{School of Mathematical Sciences, Nankai University, Tianjin 300071, P.R.China
  (\email{daopingzhang@nankai.edu.cn}).}
\and Gary P. T. Choi\thanks{Department of Mathematics, Massachusetts Institute of Technology, Cambridge, MA 02139, USA
  (\email{ptchoi@mit.edu}).}
  \and Jianping Zhang\thanks{School of Mathematics and Computational Science, Hunan National Center for Applied Mathematics, and Hunan Key Laboratory for Computation and Simulation in Science and Engineering, Xiangtan University, Xiangtan, Hunan 411105, P.R.China
  (\email{jpzhang@xtu.edu.cn}).}
  \and Lok Ming Lui\thanks{Department of Mathematics, The Chinese University of Hong Kong, Hong Kong
  (\email{lmlui@math.cuhk.edu.hk}).}}

\usepackage{amsopn}
\DeclareMathOperator{\diag}{diag}

\ifpdf
\hypersetup{
  pdftitle={A unifying framework for n-dimensional quasi-conformal mappings},
  pdfauthor={Daoping Zhang, Gary P. T. Choi, Jianping Zhang and Lok Ming Lui}
}
\fi

\begin{document}

\maketitle

% REQUIRED
\begin{abstract}
With the advancement of computer technology, there is a surge of interest in effective mapping methods for objects in higher-dimensional spaces. To establish a one-to-one correspondence between objects, higher-dimensional quasi-conformal theory can be utilized for ensuring the bijectivity of the mappings. In addition, it is often desirable for the mappings to satisfy certain prescribed geometric constraints and possess low distortion in conformality or volume. In this work, we develop a unifying framework for computing $n$-dimensional quasi-conformal mappings. More specifically, we propose a variational model that integrates quasi-conformal distortion, volumetric distortion, landmark correspondence, intensity mismatch and volume prior information to handle a large variety of deformation problems. We further prove the existence of a minimizer for the proposed model and devise efficient numerical methods to solve the optimization problem. We demonstrate the effectiveness of the proposed framework using various experiments in two- and three-dimensions, with applications to medical image registration, adaptive remeshing and shape modeling.
\end{abstract}

% REQUIRED

% REQUIRED

%%%%%%%%%%%%%%%%%%%%%%%%%%%%%%%%%%%%%%%%%%%%%%%%%%%%%%%%%%%
\section{Introduction} \label{sect:introduction}
A fundamental task in imaging science is to find an optimal transformation of certain given objects. For instance, given a pair of two-dimensional (2D) images, it is common to search for a mapping that deforms one of the images to match the other image such that the features of the two images are aligned as much as possible. Traditional mapping methods use rigid transformations, isotropic or anisotropic scaling and shear transformations, which are limited by their degrees of freedom and hence do not yield an accurate registration between corresponding objects in general. In recent decades, more advanced non-rigid mapping methods have been developed, including the thin plate splines (TPS) method~\cite{bookstein1989principal}, large deformation diffeomorphic metric mapping (LDDMM)~\cite{joshi2000landmark,beg2005computing}, Demons algorithm~\cite{thirion1998image,vercauteren2009diffeomorphic} etc. (see~\cite{zitova2003image,hernandez2008comparing} for more details). In particular, prescribed landmarks are commonly used to aid the computation of the mappings~\cite{rohr2001landmark,johnson2002consistent}. Area information has also been utilized in the computation of optimal mass transport maps~\cite{zhao2013area,gu2016variational,giri2021open} and density-equalizing maps~\cite{choi2018density,choi2020area,choi2021volumetric}. More recently, quasi-conformal mappings have become increasingly popular for the development of non-rigid image registration~\cite{lam2014landmark,yung2018efficient,tu2020diffeomorphic,zeng2011registration} and surface mapping methods~\cite{lipman2012bounded,lipman2012simple,weber2012computing,wong2014computation,choi2015fast,meng2016tempo,chien2016bounded,lam2017optimized,choi2018linear,choi2020parallelizable,wong2015computing,yang2020quasiconformal,zeng2011registration}, with applications to geometry processing~\cite{choi2016spherical,choi2016fast,choi2017subdivision}, biological shape analysis~\cite{choi2020tooth,choi2020shape} and medical visualization~\cite{choi2015flash,choi2017conformal,zeng2014colon,ta2021quantitative}. Specifically, quasi-conformal theory allows one to ensure the bijectivity and reduce the local geometric distortion of the mappings. However, most of the above-mentioned methods only work for 2D objects embedded in three-dimensional (3D) Euclidean space but not higher-dimensional shapes. While a few recent works have extended the computation of quasi-conformal mappings to higher dimensions~\cite{lee2016landmark,naitsat2018geometric,naitsat2018geometry,zhang20203d,naitsat2021inversion}, they primarily focus on the conformal distortion and landmark mismatch but not the volumetric distortion or any other useful prior information. In this work, we propose a unifying framework for computing $n$-dimensional quasi-conformal mappings, which are folding-free quasi-regular mappings in $\mathbb{R}^n$. Unlike the prior higher-dimensional mapping methods, our proposed framework considers a variational model that involves not only quasi-conformality and landmark constraints but also intensity and volumetric information. The existence of a minimizer for the variational model is theoretically guaranteed. We also introduce a novel use of an exponential term that significantly simplifies the numerical computation of the optimization problem. Altogether, our proposed framework can be effectively applied to a wide range of $n$-dimensional mapping problems. 

The organization of the paper is as follows. In Section~\ref{sect:contribution}, we highlight the contributions of our work. In Section~\ref{sect:background}, we introduce the mathematical background of our proposed framework. The detailed formulation of the proposed framework is then explained in Section~\ref{sect:main}. In Section~\ref{sect:experiment}, we demonstrate the effectiveness of the proposed framework for different $n$-dimensional mapping problems using various synthetic examples. In Section~\ref{sect:application}, we explore the applications of the proposed framework to medical image registration, adaptive remeshing and graphics. We conclude the paper and discuss possible future works in Section~\ref{sect:conclusion}.

%%%%%%%%%%%%%%%%%%%%%%%%%%%%%%%%%%%%%%%%%%%%%%%%%%%%%%%%%%%
\section{Contributions} \label{sect:contribution}
The contributions of our work are as follows:
\begin{enumerate}[(i)] 
\item Our proposed framework takes quasi-conformal distortion, volumetric distortion, landmark correspondence, intensity mismatch and volume prior into consideration, which allows us to handle a broader class of $n$-dimensional mapping problems when compared to the existing methods.
\item The existence of the minimizer for our variational model is theoretically guaranteed.
\item The bijectivity of the mappings obtained by our framework is also guaranteed.
\item The computation of the mappings is more efficient than the prior quasi-conformal mapping methods.
\item The proposed framework can be effectively applied to medical image registration, adaptive remeshing and graphics.
\end{enumerate}

%%%%%%%%%%%%%%%%%%%%%%%%%%%%%%%%%%%%%%%%%%%%%%%%%%%%%%%%%%%
\section{Mathematical background}\label{sect:background}
In this section, we review some basic mathematical concepts relevant to this work, including 2D and $n$-dimensional quasi-conformal maps.

\begin{figure}[t]
    \centering
    \includegraphics[width=\textwidth]{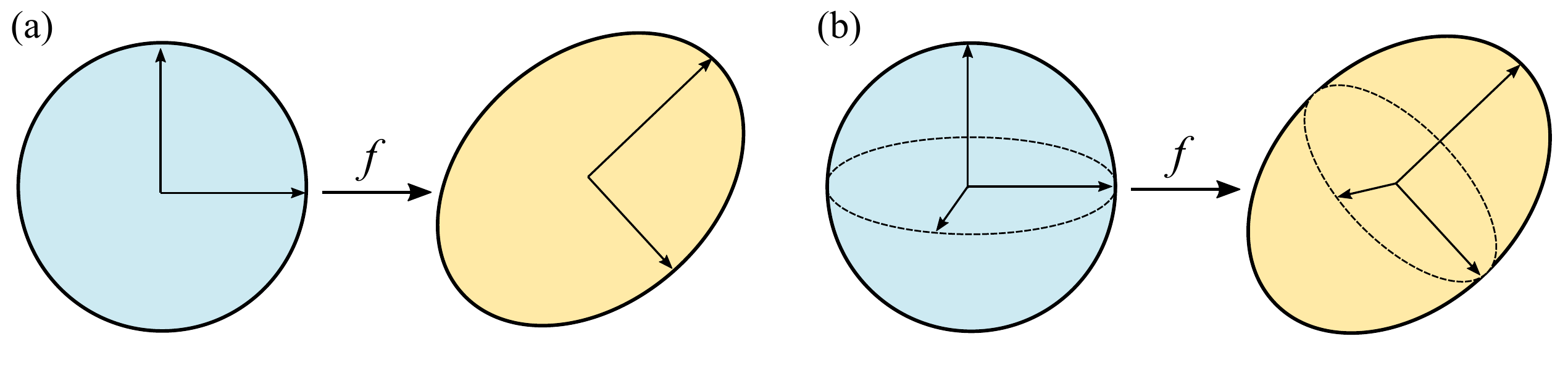}
    \caption{An illustration of quasi-conformal maps. (a) Under a 2D quasi-conformal map, infinitesimal circles are mapped to infinitesimal ellipses. (b) Under a $n$-dimensional quasi-conformal map, infinitesimal spheres of dimension $n-1$ are mapped to infinitesimal ellipsoids of dimension $n-1$.}
    \label{fig:qc_illustration}
\end{figure}

\subsection{2D quasi-conformal maps}
Mathematically, an orientation-preserving homeomorphism $f(z):\mathbb{C}\rightarrow\mathbb{C}$ is said to be \emph{quasi-conformal} if it satisfies the following Beltrami equation~\cite{Lehto1973quasiconformal}:
\begin{equation}\label{BE}
\frac{\partial f}{\partial \bar{z}} = \mu(z)\frac{\partial f}{\partial z}
\end{equation}
for some complex-valued Lebesgue measurable function $\mu:\mathbb{C}\rightarrow\mathbb{C}$ satisfying $\|\mu\|_{\infty}<1$, where $\mu$ is called the Beltrami coefficient \cite{bers1977quasiconformal} and $z= x_{1}+\bm{i}x_{2}$, where $x_1, x_2$ are real numbers. Since $\frac{\partial f}{\partial \bar{z}} = \frac{1}{2}(\frac{\partial f}{\partial x_{1}}+\bm{i}\frac{\partial f}{\partial x_{2}})$ and $\frac{\partial f}{\partial z} = \frac{1}{2}(\frac{\partial f}{\partial x_{1}}-\bm{i}\frac{\partial f}{\partial x_{2}})$, we have 
\begin{equation}\label{NBC}
|\mu(z)|^{2}  = \frac{\|\nabla f\|_{\mathrm{F}}^{2}-2\det\nabla f}{\|\nabla f\|_{\mathrm{F}}^{2}+2\det\nabla f},
\end{equation}
where $\nabla f$ is the Jacobian of $f$ and $\|\cdot\|_{\mathrm{F}}$ represents the Frobenius norm. From \eqref{NBC}, we can see that $\|\mu(z)\|_{\infty}<1 \Longleftrightarrow \det\nabla f >0$, which means that a quasi-conformal mapping is indeed one-to-one. From $\mu(z)$, we can know that the angle of maximal magnification is $\mathrm{arg}(\mu(z))/2$ with magnifying factor $\left|\frac{\partial f}{\partial z}\right| (1+|\mu(z)|)$, the angle of maximal shrinking is the orthogonal angle $(\mathrm{arg}(\mu(z))-\pi)/2$ with shrinking factor $\left|\frac{\partial f}{\partial z}\right|(1-|\mu(z)|)$~\cite{lam2014landmark}, which means that the Beltrami coefficient $\mu$ represents the local geometric distortion (see Fig.~\ref{fig:qc_illustration}(a)). Furthermore, based on the Beltrami coefficient, we can define by $K_{d}$ the dilatation
\begin{equation}\label{dilatation}
K_{d}(f) = \frac{1+|\mu|}{1-|\mu|}
\end{equation}
to express the ratio of the largest singular value of the Jacobian of $f$ divided by the smallest singular value. In addition, if $\mu(z) = 0$, we can get the complex form of the Cauchy-Riemann equation:
\begin{equation}
\bm{i}\frac{\partial f}{\partial x_{1}}=\frac{\partial f}{\partial x_{2}},
\end{equation}
which shows that $f$ is a conformal map. From this perspective, quasi-conformal maps are a generalization of conformal maps. 

\subsection{$n$-dimensional quasi-conformal maps}
Since the Beltrami coefficient is only defined in the complex space, we cannot directly define the $n$-dimensional quasi-conformal maps by extending the Beltrami coefficient or the dilatation $K_d$ to the $n$-dimensional Euclidean space with $n \geq 3$. Nevertheless, one can generalize the notion of conformality distortion for a diffeomorphism in the $n$-dimensional space as explained below. 

An orientation-preserving mapping $\bm{f}(\bm{x}):\mathbb{R}^{n}\rightarrow\mathbb{R}^{n}$ is conformal \cite{gehring2017introduction,iwaniec2001geometric} if
\begin{equation}\label{ndCon}
\nabla\bm{f}^{T}\nabla\bm{f} = (\det\nabla\bm{f})^{2/n}I,
\end{equation}
where $\nabla\bm{f}$ represents the Jacobian of the mapping $\bm{f}$ and $I$ is the identity matrix.
Set $\lambda_{1}\geq\dots\geq\lambda_{n}>0$ as the eigenvalues of $\nabla\bm{f}^{T}\nabla\bm{f}$. 
Then by the eigendecomposition of $\nabla\bm{f}^{T}\nabla\bm{f}$, it is easy to verify that \eqref{ndCon} holds if and only if $\lambda_{1} = \dots = \lambda_{n}$. In addition, on one hand, by the inequality of arithmetic and geometric means, 
\begin{equation}\label{AGM}
\frac{1}{n} \sum_{i=1}^{n}\lambda_{i} \geq (\Pi_{i=1}^{n}\lambda_{i})^{1/n},
\end{equation}
we have 
\begin{equation}
\frac{1}{n}\left(\frac{\|\nabla\bm{f}\|_{\mathrm{F}}^{2}}{(\det\nabla\bm{f})^{2/n}}\right) \geq 1,
\end{equation}
since $\det\nabla\bm{f} = (\Pi_{i=1}^{n}\lambda_{i})^{1/2}$ and $\|\nabla\bm{f}\|_{\mathrm{F}}^{2} = \sum_{i=1}^{n}\lambda_{i}$. On the other hand, since the equality holds in \eqref{AGM} if and only if $\lambda_{1} = \dots = \lambda_{n}$, based on the above discussion, we can see that $\bm{f}$ is conformal if and only if $\frac{1}{n}\left(\frac{\|\nabla\bm{f}\|_{\mathrm{F}}^{2}}{(\det\nabla\bm{f})^{2/n}}\right) = 1$. Hence, the quantity $\frac{1}{n}\left(\frac{\|\nabla\bm{f}\|_{\mathrm{F}}^{2}}{(\det\nabla\bm{f})^{2/n}}\right)$ can be regarded as a measure of how far away $\bm{f}$ is from a conformal map. In \cite{lee2016landmark}, based on this idea, a generalized conformality distortion $K(\bm{f})$ in the $n$-dimensional space is defined:
\begin{equation}\label{nDistortion}
K(\bm{f}):= \left\{
\begin{split}
& \frac{1}{n}\left(\frac{\|\nabla\bm{f}\|_{\mathrm{F}}^{2}}{(\det\nabla\bm{f})^{2/n}}\right) \quad \mathrm{if} \ \det\nabla\bm{f}>0,\\
& +\infty \quad \quad \quad \qquad \qquad \mathrm{otherwise}.
\end{split}\right.
\end{equation}
Furthermore, for $n=2$, it has been pointed out in \cite{zhang20203d} that 
\begin{equation}
K(\bm{f})\leq K_{d}(\bm{f}) \leq 2K(\bm{f}),
\end{equation}
which shows that $K(\bm{f})$ can be used to measure the dilatation of a $n$-dimensional quasi-conformal map (see also Fig.~\ref{fig:qc_illustration}(b)).

%%%%%%%%%%%%%%%%%%%%%%%%%%%%%%%%%%%%%%%%%%%%%%%%%%%%%%%%%%%
\section{Proposed framework}\label{sect:main}
In this section, we first introduce our proposed unifying framework for $n$-dimensional quasi-conformal mappings. We then devise the numerical implementation for solving the proposed variational model.

\subsection{A new framework for $n$-dimensional deformation}
Before introducing our proposed framework, we first review a related work as the motivation of our proposed framework. As described in Section \ref{sect:background}, one can define a generalized conformality distortion $K$ as in \eqref{nDistortion} to measure how far away a map is from a conformal map in the $n$-dimensional space. Based on this term $K$, Lee \textit{et al.} proposed a landmark-matching model \cite{lee2016landmark}:
\begin{equation}\label{LLL}
\min_{\bm{y}} \int_{\Omega}K(\bm{y})\mathrm{d}\bm{x} + \frac{\alpha}{2}\int_{\Omega}\|\Delta\bm{y}\|^{2}_{\mathrm{F}}\mathrm{d}\bm{x}, \quad \mathrm{s.t.} \quad \bm{y}(p_{i}) = q_{i},\ i=1,\dots,m,    
\end{equation}
where $\bm{y}:\mathbb{R}^{n}\rightarrow\mathbb{R}^{n}$ is the desired transformation, $\alpha>0$ is a positive parameter to balance the distortion term and the smooth term, and $p_{i}$ and $q_{i}$ are prescribed landmark points. We can see that minimizing the model \eqref{LLL} with a suitable parameter $\alpha$ leads to a smooth transformation that simultaneously achieves a minimal conformality distortion and satisfies the landmark constraints. In the 3D case, to solve the model \eqref{LLL}, an auxiliary variable $\bm{v}$ is introduced and an equivalent formulation can be derived:
\begin{equation}\label{LLLequivalent}
\min_{\bm{y},\bm{v}} \int_{\Omega}\frac{\|\nabla\bm{y}\|_{\mathrm{F}}^{2}}{(3\det\bm{v})^{2/3}}\mathrm{d}\bm{x} + \frac{\alpha}{2}\int_{\Omega}\|\Delta\bm{y}\|^{2}_{\mathrm{F}}\mathrm{d}\bm{x}, \ \mathrm{s.t.} \ \bm{v}=\nabla\bm{y},\ \det\bm{v}>0,\ \bm{y}(p_{i}) = q_{i}, i=1,\dots,m.   
\end{equation}
Then the alternating direction method of multipliers (ADMM) can be employed to solve \eqref{LLLequivalent}. However, note that the model \eqref{LLL} does not incorporate any intensity or volume information which could improve the accuracy of the mappings in registration problems. In addition, one subproblem in ADMM for solving \eqref{LLLequivalent} involves the inequality constraint $\det\bm{v}>0$, which makes solving the problem computationally expensive.

To overcome these two issues, we first propose the following model for computing $n$-dimensional quasi-conformal mappings:
\begin{equation}\label{PM1}
\begin{split}
&\min_{\bm{y},\theta} \frac{\alpha_{1}}{2}\int_{\Omega}|\theta|^{2}\mathrm{d}\bm{x}+\alpha_{2}\int_{\Omega}\frac{\|\nabla\bm{y}\|_{\mathrm{F}}^{2}}{n(\det\nabla\bm{y})^{2/n}}\mathrm{d}\bm{x} + \frac{\alpha_{3}}{2}\int_{\Omega}\|\Delta\bm{y}\|^{2}_{2}\mathrm{d}\bm{x}+\frac{\alpha_{4}}{2}\int_{\Omega}(T\circ\bm{y}-R)^{2}\mathrm{d}\bm{x}, \\
&\mathrm{s.t.} \ \det\nabla\bm{y}=\mathrm{e}^{\theta},\ \bm{y}(p_{i}) = q_{i},\ i=1,\dots,m,   
\end{split}
\end{equation}
where $T:\Omega\subset\mathbb{R}^{n}\rightarrow \mathbb{R}$ is the template, $R:\Omega\subset\mathbb{R}^{n}\rightarrow \mathbb{R}$ is the reference, $T\circ\bm{y}$ is the deformed template and $\theta$ is a function from $\Omega\subset\mathbb{R}^{n}$ to $\mathbb{R}$. When compared to the previous model \eqref{LLL}, the model \eqref{PM1} contains a new sum of squared differences (SSD). One can also see that the inequality constraint $\det\bm{v}>0$ is converted into an equality constraint by using the exponential function $\mathrm{e}^{\theta}$. Here we also add a regularizer about $\theta$ to minimize the volume change because the geometric meaning of the Jacobian determinant of the transformation represents the ratio of the change of the volume. Furthermore, if we know the volume prior in some specific region $\Omega'$, then we can combine this information with model \eqref{PM1}. This leads to the main proposed unifying model for $n$-dimensional quasi-conformal mappings in this paper:
\begin{equation}\label{PM}
\begin{split}
&\min_{\bm{y},\theta} \frac{\alpha_{1}}{2}\int_{\Omega}|\theta|^{2}\mathrm{d}\bm{x} +\alpha_{2}\int_{\Omega}\frac{\|\nabla\bm{y}\|_{\mathrm{F}}^{2}}{n(\det\nabla\bm{y})^{2/n}}\mathrm{d}\bm{x}  +  \frac{\alpha_{3}}{2}\int_{\Omega}\|\Delta\bm{y}\|^{2}_{2}\mathrm{d}\bm{x} \\
& \qquad \qquad \qquad \qquad  \qquad \qquad\qquad \qquad +\frac{\alpha_{4}}{2}\int_{\Omega'}(\theta-\bar{\theta})^{2}\mathrm{d}\bm{x}
+\frac{\alpha_{5}}{2}\int_{\Omega}(T\circ\bm{y}-R)^{2}\mathrm{d}\bm{x}, \\
&\mathrm{s.t.} \ \det\nabla\bm{y}=\mathrm{e}^{\theta},\ \bm{y}(p_{i}) = q_{i},\ i=1,\dots,m,   
\end{split}
\end{equation}
where $\bar{\theta}$ is a given function indicating the volume prior in the specific region $\Omega'$.

Next, we prove the existence of the solution for the proposed model \eqref{PM}.
\begin{theorem}\label{thm:existence}
Suppose $\Omega$ is bounded and simply connected, $T, R$ are continuous functions from $\Omega\subset\mathbb{R}^{n}\rightarrow\mathbb{R}$, and $\alpha_{i} >0, i=1,\dots,5$. Let
\begin{equation}
\begin{split}
\mathcal{A}:= \{\bm{y}\in \mathcal{C}^{2}(\Omega),\theta\in \mathcal{C}^{0}(\Omega): & \|\theta\|_{\infty}\leq c_{1}, \|\bm{y}\|_{\infty}\leq c_{2}, \|\nabla\bm{y}\|_{\infty}\leq c_{3},\\
& \|\nabla^{2}\bm{y}\|_{\infty}\leq c_{4}, \det\nabla\bm{y}=\mathrm{e}^{\theta}, \bm{y}(p_{i}) = q_{i}, i=1,\dots,m\}
\end{split}
\end{equation}
for some $c_{i} >0, i=1,\dots,4$. Then the proposed model \eqref{PM} admits a minimizer in $\mathcal{A}$. In fact, $\mathcal{A}$ is compact.
\end{theorem}

\begin{proof}
By the Homogeneity Lemma, there exists a $\mathcal{C}^{\infty}$ mapping $\bm{y}^{*}$ satisfying the landmarks constraint \cite{joshi2000landmark}, which is also orientation-preserving. Then we can find the corresponding $\theta^{*}$ such that $\det\nabla\bm{y}^{*} = \mathrm{e}^{\theta^{*}}$. Here, it is clear that $\theta^{*}$ belongs to $\mathcal{C}^{0}(\Omega)$. Hence, by setting $c_{1} = \|\theta^{*}\|_{\infty} , c_{2} = \|\bm{y}^{*}\|_{\infty}, c_{3} = \|\nabla\bm{y}^{*}\|_{\infty}$ and $c_{4} = \|\nabla^{2}\bm{y}^{*}\|_{\infty}$, we have $(\bm{y}^{*},\theta^{*})$ in $\mathcal{A}$, which indicates that $\mathcal{A}$ is non-empty.

To prove that $\mathcal{A}$ is compact, we need to check that $\mathcal{A}$ is complete and totally bounded.

Define $\|\bm{y}\|_{y}:= \|\bm{y}\|_{\infty} + \|\nabla\bm{y}\|_{\infty} + \|\nabla^{2}\bm{y}\|_{\infty} $
and $\|(\bm{y},\theta)\|_{s} := \|\theta\|_{\infty} + \|\bm{y}\|_{y}$. Let the sequence $\{\bm{y}^{k},\theta^{k}\}_{k=1}^{\infty}$ in $\mathcal{A}$ be the Cauchy sequence with respect to the norm $\|\cdot\|_{s}$. Then $\{\theta^{k}\}_{k=1}^{\infty}$ is also a  Cauchy sequence. Since $\mathcal{A}_{\theta} := \{\theta\in\mathcal{C}^{0}(\Omega): \|\theta\|\leq c_{1}\}$ is complete with respect to $\|\cdot\|_{\infty}$, there exists $\bar{\theta}$ such that $\theta^{k} \rightarrow \bar{\theta}$. Similarly, there exist $\bar{\bm{y}}, \bm{u}$ and $\bm{v}$ such that $\bm{y}^{k}\rightarrow \bar{\bm{y}}$,  $\nabla\bm{y}^{k}\rightarrow \bm{u}$ and  $\nabla^{2}\bm{y}^{k}\rightarrow \bm{v}$ respectively. Furthermore, since $\bar{\bm{y}}$ is $\mathcal{C}^{2}$, we have $\bm{u} = \nabla\bar{\bm{y}}$ and $\bm{v} = \nabla^{2}\bar{\bm{y}}$. So  $\mathcal{A}_{\bm{y}} := \{\bm{y}\in\mathcal{C}^{2}(\Omega): \|\bm{y}\|_{\infty}\leq c_{2}, \|\nabla\bm{y}\|_{\infty}\leq c_{3},\|\nabla^{2}\bm{y}\|_{\infty}\leq c_{4}\}$ is complete with respect to $\|\cdot\|_{y}$. Hence, we have $(\bm{y}^{k},\theta^{k})\rightarrow (\bar{\bm{y}},\bar{\theta})$. In addition, since $\det\nabla\bm{y}^{k} = \mathrm{e}^{\theta^{k}}$ and $\bm{y}^{k}(p_{i}) = q_{i}$, by the continuity, we have $\det\nabla\bar{\bm{y}} = \mathrm{e}^{\bar{\theta}}$ and $\bar{\bm{y}}(p_{i}) = q_{i}$, which shows that $(\bar{\bm{y}},\bar{\theta})$ is in $\mathcal{A}$. Hence, $\mathcal{A}$ is complete with respect to the norm $\|\cdot\|_{s}$.

To show that $\mathcal{A}$ is totally bounded, we prove that the product space $\mathcal{A}_{\bm{y}}\times \mathcal{A}_{\theta}$, which contains $\mathcal{A}$, is totally bounded. Define a regular grid on $\Omega$ with edge lengths $\frac{1}{m_{1}}$. Denote the grid points by $\{\bm{x}_{i}\}_{i\in I}$ after re-indexing. We define the set of tent functions on each grid point by $\mathcal{T}_{m_{1},n_{1}} = \{\frac{k}{n_{1}}\phi_{\bm{x}_{i}}\}_{i\in I, 1\leq k\leq n_{1}}$. Here, $\frac{1}{n_{1}}$ is the length of each interval dividing $[0,1]$, and $\phi_{\bm{x}_{i}}$ is the tent function such that $\phi_{\bm{x}_{i}}(\bm{x}_{i}) = c_{1}$ and $\phi_{\bm{x}_{i}}(\bm{x}_{j}) = 0$ for $i\neq j$. Let $\mathcal{B}_{m_{1},n_{1}} = \{\hat{\theta}\in\mathcal{A}_{\theta}: \hat{\theta}=\sum_{i}\mathcal{F}_{i}, \mathcal{F}_{i}\in\mathcal{T}_{m_{1},n_{1}}\}$, which has finitely many elements. In this way, given any $\epsilon>0$, we can choose large enough $m_{1},n_{1}$ such that for any $\theta\in\mathcal{A}_{\theta}$, there exists a $\hat{\theta}\in\mathcal{B}_{m_{1},n_{1}}$ with $\|\theta-\hat{\theta}\|_{\infty}<\epsilon$. Hence, $\mathcal{A}_{\theta}\subset\cup_{\hat{\theta}\in\mathcal{B}_{m_{1},n_{1}}} B_{\epsilon}(\hat{\theta})$, where $B_{\epsilon}(\hat{\theta})$ is the open ball centering at $\hat{\theta}$ with radius $\epsilon$. Therefore, $\mathcal{A}_{\theta}$ is totally bounded. In fact, each element in $\mathcal{T}_{m_{1},n_{1}} = \{\frac{k}{n_{1}}\phi_{\bm{x}_{i}}\}_{i\in I, 1\leq k\leq n_{1}}$ is a $\mathcal{C}^{0}$ spline. Following this methodology, we can construct $\mathcal{P}_{m_{2},n_{2}}$ containing finite $\mathcal{C}^{2}$ splines such that for given any $\epsilon>0$ and any $\bm{y}\in\mathcal{A}_{\bm{y}}$, for large enough $m_{2}$ and $n_{2}$, there exists a $\hat{\bm{y}} \in \mathcal{B}_{m_{2},n_{2}} = \{\hat{\bm{y}}\in\mathcal{A}_{\bm{y}}: \hat{\bm{y}}=\sum_{i}\mathcal{F}_{i}, \mathcal{F}_{i}\in\mathcal{T}_{m_{2},n_{2}}\}$ such that $\|\bm{y} - \hat{\bm{y}}\|_{y} < \epsilon$. Then we can conclude that $\mathcal{A}_{\bm{y}}$ is also totally bounded. Hence, for given any $\epsilon>0$, by setting $m=\max\{m_{1},m_{2}\}$ and $n=\max\{n_{1},n_{2}\}$, we see that $\mathcal{A}_{\bm{y}}\times \mathcal{A}_{\theta}$ is totally bounded. Since any subset of a totally bounded set is totally bounded, we conclude that $\mathcal{A}$ is totally bounded. 

Therefore, $\mathcal{A}$ is compact. Since the objective functional in the proposed model \eqref{PM} is continuous in $\mathcal{A}$, the proposed model admits a minimizer in $\mathcal{A}$.
\end{proof}

Note that the proposed unifying model \eqref{PM} contains different components that correspond to different geometric quantities. By suitably setting the parameters of the components, we can obtain different models that yield different types of $n$-dimensional quasi-conformal mappings. Below, we list some possible formulations and explain their usage.

\begin{enumerate}

\item \textbf{Landmark-constrained $n$-dimensional quasi-conformal mapping}.  Let $\alpha_{1} =\alpha_{4} = \alpha_{5} = 0$. The proposed model \eqref{PM} becomes
\begin{equation}\label{case1}
\begin{split}
&\min_{\bm{y},\theta} \alpha_{2}\int_{\Omega}\frac{\|\nabla\bm{y}\|_{\mathrm{F}}^{2}}{n(\det\nabla\bm{y})^{2/n}}\mathrm{d}\bm{x}  +  \frac{\alpha_{3}}{2}\int_{\Omega}\|\Delta\bm{y}\|^{2}_{2}\mathrm{d}\bm{x}, \\
&\mathrm{s.t.} \ \det\nabla\bm{y}=\mathrm{e}^{\theta},\ \bm{y}(p_{i}) = q_{i},\ i=1,\dots,m.  
\end{split}
\end{equation}
Note that \eqref{case1} is theoretically equivalent to the formulation \eqref{LLL} proposed in~\cite{lee2016landmark}. However, the use of the exponential function $\mathrm{e}^{\theta}$ in \eqref{case1} largely simplifies the numerical implementation and hence leads to a significant improvement in practice when compared to \eqref{LLL} (see Section~\ref{implementation} for more details). 

\item \textbf{Landmark- and intensity-based $n$-dimensional quasi-conformal registration}. Let $\alpha_{1} = \alpha_{4} = 0$. The proposed model \eqref{PM} becomes
\begin{equation}\label{case2}
\begin{split}
&\min_{\bm{y},\theta} \alpha_{2}\int_{\Omega}\frac{\|\nabla\bm{y}\|_{\mathrm{F}}^{2}}{n(\det\nabla\bm{y})^{2/n}}\mathrm{d}\bm{x}  +  \frac{\alpha_{3}}{2}\int_{\Omega}\|\Delta\bm{y}\|^{2}_{2}\mathrm{d}\bm{x} +\frac{\alpha_{5}}{2}\int_{\Omega}(T\circ\bm{y}-R)^{2}\mathrm{d}\bm{x}, \\
&\mathrm{s.t.} \ \det\nabla\bm{y}=\mathrm{e}^{\theta},\ \bm{y}(p_{i}) = q_{i},\ i=1,\dots,m.  
\end{split}
\end{equation}
The model can be viewed as an extension of~\cite{lam2014landmark,lee2016landmark}. It is a hybrid registration model based on the intensity information and landmarks, which simultaneously guarantees the bijectivity of the resulting transformation.

\item \textbf{$n$-dimensional quasi-conformal mapping with volume prior and optimized volumetric distortion}. Let $\alpha_{5}=0$ and suppose there is no landmark. The proposed model \eqref{PM} becomes 
\begin{equation}\label{case3}
\begin{split}
&\min_{\bm{y},\theta} \frac{\alpha_{1}}{2}\int_{\Omega}|\theta|^{2}\mathrm{d}\bm{x} +\alpha_{2}\int_{\Omega}\frac{\|\nabla\bm{y}\|_{\mathrm{F}}^{2}}{n(\det\nabla\bm{y})^{2/n}}\mathrm{d}\bm{x}  +  \frac{\alpha_{3}}{2}\int_{\Omega}\|\Delta\bm{y}\|^{2}_{2}\mathrm{d}\bm{x}  +\frac{\alpha_{4}}{2}\int_{\Omega'}(\theta-\bar{\theta})^{2}\mathrm{d}\bm{x}, \\
&\mathrm{s.t.} \ \det\nabla\bm{y}=\mathrm{e}^{\theta}.   
\end{split}
\end{equation}
Note that the resulting transformation will be as volume-preserving as possible because of the first term, while the specific region $\Omega'$ will change in volume based on the volume prior.

\item \textbf{The most general model}. In case all components are needed, we can set $\alpha_i \neq 0$ for all $i = 1, \dots, 5$ to keep all terms in the proposed model \eqref{PM} and solve it directly.

\end{enumerate}

\begin{remark}
While the theoretical minimizer in Theorem~\ref{thm:existence} is a global minimizer of the energy~\eqref{PM}, in general it may not fully satisfy the prescribed constraints. Note that the volume prior and intensity constraints can be set arbitrarily, and in some cases they may not even be compatible with each other and/or the landmark constraints. Therefore, the minimizer will just give the lowest energy by considering all components and achieving a balance between them. If a certain constraint is desired to be more important, we can set a larger weight for it in the unifying model~\eqref{PM}. As the weight tends to infinity, the constraint violation for the corresponding term will go to 0 in the resulting minimizer. 
\end{remark}

\begin{remark}
One may also be interested in the uniqueness of the solution to the proposed model~\eqref{PM}.
To analyze it, we can consider the components in the model one by one. Note that the first term and the fourth term involving $\theta$ are invariant under translations and rotations. The second term involving the generalized conformality distortion is invariant under M\"obius transformations, analogous to the 2D quasi-conformal case. The third term involving the second order regularization is invariant under translations, scalings and rotations. However, the fifth term and the landmark constraints may not be invariant under translations, scalings or rotations. This shows that the M\"obius transformation freedom is only in certain components of the proposed model. If the fitting term and the landmark constraints are not considered, the solution of the model is not unique, while if we have a general model with the presence of all components, there may not be any above-mentioned freedom in the solution. 
\end{remark}

\subsection{Implementation} \label{implementation}
In this section, we devise a numerical method to solve the model \eqref{PM} using ADMM.

Firstly, we review the main idea of ADMM. Consider the following general optimization problem:
\begin{equation}\label{Example}
\min_{u,v} h(u)+g(v) \ \mathrm{s.t.} \ c(u) = v,
\end{equation}
where $h:\mathcal{U}\rightarrow \mathbb{R}$, $g:\mathcal{V}\rightarrow \mathbb{R}$ and $c:\mathcal{U}\rightarrow \mathcal{V}$ are three functions. The augmented Lagrangian function of \eqref{Example} is as follows:
\begin{equation}\label{Example_ALM}
\mathcal{L}(u,v,\lambda,\rho) := h(u)+g(v) + \langle\lambda, c(u)-v\rangle + \frac{\rho}{2}\|c(u)-v\|_{2}^{2},
\end{equation}
where $\lambda\in\mathcal{V}$ is the Lagrange multiplier, $\langle\cdot,\cdot\rangle$ is the corresponding inner product and $\rho>0$ is a penalty parameter. ADMM aims to solve \eqref{Example_ALM} using an iterative scheme. In particular, the $k$-th iteration of ADMM is as follows:

\begin{equation}
\left\{
\begin{split}
u^{k+1} &:= \mathrm{argmin}_{u}\ h(u)+\frac{\rho}{2}\left\|c(u)-v^{k}+\frac{\lambda^{k}}{\rho}\right\|_{2}^{2},\\
v^{k+1} &:= \mathrm{argmin}_{v}\ g(v)+\frac{\rho}{2}\left\|c(u^{k+1})-v+\frac{\lambda^{k}}{\rho}\right\|_{2}^{2}, \\
\lambda^{k+1} &:= \lambda^{k}+\rho(c(u^{k+1})-v^{k+1}).
\end{split}
\right.
\end{equation}
In other words, ADMM first solves for $u^{k+1}$ by fixing $v=v^{k}$ and then solves for $v^{k+1}$ by fixing $u=u^{k+1}$.

Now, we are ready to apply ADMM to the proposed model \eqref{PM}. First, we write the corresponding augmented Lagrangian function of \eqref{PM}:
\begin{equation}\label{ALF}
\begin{split}
\mathcal{L}(\bm{y},\theta,\lambda_{1},\lambda_{2}, \rho_{1}, \rho_{2}) &:= \frac{\alpha_{1}}{2}\int_{\Omega}|\theta|^{2}\mathrm{d}\bm{x}+\alpha_{2}\int_{\Omega}\frac{\|\nabla\bm{y}\|_{\mathrm{F}}^{2}}{n(\mathrm{e}^{\theta})^{2/n}}\mathrm{d}\bm{x}+\frac{\alpha_{3}}{2}\int_{\Omega}\|\Delta\bm{y}\|_{\mathrm{F}}^{2}\mathrm{d}\bm{x}\\
&+\frac{\alpha_{4}}{2}\int_{\Omega'}(\theta-\bar{\theta})^{2}\mathrm{d}\bm{x}+\frac{\alpha_{5}}{2}\int_{\Omega}(T\circ\bm{y} - R)^{2}\mathrm{d}\bm{x}  \\
&+ \int_{\Omega}\langle \lambda_{1}, \det\nabla\bm{y}-\mathrm{e}^{\theta} \rangle \mathrm{d}\bm{x} + \frac{\rho_{1}}{2}\int_{\Omega}(\det\nabla\bm{y}-\mathrm{e}^{\theta})^{2}\mathrm{d}\bm{x} \\
&+ \sum_{i=1}^{m} (\lambda_{2}^{i})^{T}(\bm{y}(p_{i})-q_{i}) + \frac{\rho_{2}}{2}\sum_{i=1}^{m}|\bm{y}(p_{i})-q_{i}|^{2}.
\end{split}
\end{equation}
Here, $\lambda_{2}^{i}$ is the $i$-th component of $\lambda_{2}$. Also, for the second term, we can replace $\det\nabla\bm{y}$ with $\mathrm{e}^{\theta}$, which helps simplify one subproblem in ADMM without changing the original problem. Hence, we have the following iterative scheme:
\begin{equation}\label{IteraiveScheme}
\left\{\begin{aligned}
\theta^{k+1}  :=&\ \mathrm{argmin}_{\theta}\ \frac{\alpha_{1}}{2}\int_{\Omega}|\theta|^{2}\mathrm{d}\bm{x} +\alpha_{2}\int_{\Omega}\frac{\|\nabla\bm{y}^{k}\|_{\mathrm{F}}^{2}}{n(\mathrm{e}^{\theta})^{2/n}}\mathrm{d}\bm{x}+ \frac{\alpha_{4}}{2}\int_{\Omega'}(\theta-\bar{\theta})^{2}\mathrm{d}\bm{x} \\
&+ \frac{\rho_{1}}{2}\int_{\Omega}\left(\det\nabla\bm{y}^{k}-\mathrm{e}^{\theta}+\frac{\lambda_{1}^{k}}{\rho_{1}}\right)^{2}\mathrm{d}\bm{x}, \\
\bm{y}^{k+1}  :=&\ \mathrm{argmin}_{\bm{y}}\ \alpha_{2}\int_{\Omega}\frac{\|\nabla\bm{y}\|_{\mathrm{F}}^{2}}{n(\mathrm{e}^{\theta^{k+1}})^{2/n}}\mathrm{d}\bm{x}+\frac{\alpha_{3}}{2}\int_{\Omega}\|\Delta\bm{y}\|_{\mathrm{F}}^{2}\mathrm{d}\bm{x}+\frac{\alpha_{5}}{2}\int_{\Omega}(T\circ\bm{y} - R)^{2}\mathrm{d}\bm{x} \\
&+\frac{\rho_{1}}{2}\int_{\Omega}\left(\det\nabla\bm{y}-\mathrm{e}^{\theta^{k+1}}+\frac{\lambda_{1}^{k}}{\rho_{1}}\right)^{2}\mathrm{d}\bm{x} + \frac{\rho_{2}}{2}\sum_{i=1}^{m} \left|\bm{y}(p_{i})-q_{i}+\frac{(\lambda_{2}^{i})^{k}}{\rho_{2}}\right|, \\
\lambda_{1}^{k+1} :=&\ \lambda_{1}^{k} + \rho_{1}(\det\nabla\bm{y}^{k+1}-\mathrm{e}^{\theta^{k+1}}), \\
(\lambda_{2}^{i})^{k+1} :=&\ (\lambda_{2}^{i})^{k}+ \rho_{2}(\bm{y}^{k+1}(p_{i})-q_{i}), \ i = 1,\dots,m.
\end{aligned}\right.
\end{equation}

Next, we discretize \eqref{IteraiveScheme} and choose the suitable optimization methods to solve the resulting finite dimensional optimization subproblems.

\subsubsection{Discretization}
For simplicity, here we describe the discretization for the case $n=2$, which can be easily extended to other $n$. Let the domain $\Omega$ be $[0,1]^{2}$. By using the standard triangular partition, we divide $\Omega$ into $2N^{2}$ small triangles and have $\Omega = \cup_{i=1}^{2N^{2}}\Omega_{i}$. Let $X, Y\in\mathbb{R}^{2(N+1)^{2}}$ be the discretized identity map and discretized transformation on the nodal grid respectively. We can then use the piecewise linear function to approximate the transformation. Let $\Theta\in\mathbb{R}^{2N^{2}}$ be the discretized $\theta$ and $\Lambda\in\mathbb{R}^{2N^{2}}$ be the discretized $\lambda$ on each subdomain $\Omega_{i}$, respectively. Set $\Lambda_{2}\in\mathbb{R}^{2m}$ as the discretized $\lambda_{2}$, and $h$ as the area of the standard triangle. For the intensity term, we assume that the intensity values are defined on the cell-centered grid. Hence, we need to give an averaging matrix $P$ from the nodal grid $Y$ to the cell-centered grid $PY$ \cite{haber2004numerical,haber2007image}. Then we can set $\vec T(PY)$ and $\vec R \in \mathbb{R}^{N^{2}}$ as the discretized deformed template image and discretized reference image, respectively.

\begin{remark}
Since $Y$ is usually not located exactly at the pixel points, an interpolation operator is necessary. Here, we choose the cubic spline interpolation \cite{modersitzki2009fair} to compute $\vec T(PY)$. The linear interpolation cannot be applied because it is not differentiable at grid points.
\end{remark}

Since the approximated first order partial derivatives are constants in each subdomain $\Omega_{i}$, we can construct $A_{i}, l=1,\dots,4$ as the discrete first order operator and $B$ as the discrete Laplace operator \cite{zhang2018novel,zhang20203d}, namely,
\begin{equation}
\partial_{x_{1}}\bm{y}_{1} \approx A_{1}Y,\ \partial_{x_{2}}\bm{y}_{1}\approx A_{2}Y,\ \partial_{x_{1}}\bm{y}_{2}\approx A_{3}Y,\ \partial_{x_{2}}\bm{y}_{2}\approx A_{4}Y,\ \Delta\bm{y}\approx BY.
\end{equation}
Hence, by denoting $\odot$ as the Hadamard product, we have
\begin{equation}
\det\nabla\bm{y} \approx A_{1}Y\odot A_{4}Y - A_{2}Y\odot A_{3}Y\ \mathrm{and} \ \|\nabla\bm{y}\|_{\mathrm{F}}^{2} \approx\sum_{l=1}^{4} A_{l}Y\odot A_{l}Y. 
\end{equation}

Now, for \eqref{IteraiveScheme}, we have the following discretized scheme:
\begin{equation}\label{SIteraiveScheme}
\resizebox{\textwidth}{!}{$
\left\{\begin{aligned}
\Theta^{k+1}  :=&\ \mathrm{argmin}_{\Theta}\ J_{1}(\Theta) = \frac{\alpha_{1}h}{2}|\Theta|_{2}^{2}+\frac{\alpha_{2}h}{2}|r^{k}./\mathrm{e}^{\Theta}|_{1} + \frac{\alpha_{4}h}{2}|I_{1}\Theta-\bar{\Theta}|_{2}^{2}+\frac{\rho_{1} h}{2}|\mathrm{e}^{\Theta}-s^{k}|_{2}^{2}, \\
Y^{k+1}  :=&\ \mathrm{argmin}_{Y}\ J_{2}(Y) = \frac{\alpha_{2}h}{2}|(\sum_{l=1}^{4} A_{l}Y\odot A_{l}Y)./\mathrm{e}^{\Theta^{k+1}}|_{1}+\frac{\alpha_{3}h}{2}|BY|_{2}^{2}+\frac{\alpha_{5}h}{2}|\vec T(PY)-\vec R|_{2}^{2} \\
&+\frac{\rho_{1} h}{2}\left|A_{1}Y\odot A_{4}Y - A_{2}Y\odot A_{3}Y - \mathrm{e}^{\Theta^{k+1}} +\frac{\Lambda_{1}^{k}}{\rho_{1}}\right|_{2}^{2} + \frac{\rho_{2}}{2}\left|I_{2}Y-Q+\frac{\Lambda_{2}}{\rho_{2}}\right|^{2}, \\
\Lambda_{1}^{k+1} :=&\ \Lambda^{k}_{1} + \rho_{1}(A_{1}Y^{k+1}\odot A_{4}Y^{k+1} - A_{2}Y^{k+1}\odot A_{3}Y^{k+1}-\mathrm{e}^{\Theta^{k+1}}), \\
\Lambda_{2}^{k+1} :=&\ \Lambda^{k}_{2} + \rho_{2}(I_{2}Y^{k+1}-Q),
\end{aligned}\right.
$}
\end{equation}
where $r^{k} = \sum_{l=1}^{4}A_{l}Y^{k}\odot A_{l}{Y}^{k}$, $s^{k} = A_{1}Y^{k}\odot A_{4}Y^{k} - A_{2}Y^{k}\odot A_{3}Y^{k}+\frac{\Lambda_{1}^{k}}{\rho_{1}}$, $\bar{\Theta}$ is the discretization of the prior $\bar{\theta}$, $Q$ is the vectorization of $\bm{q}_{i}, i=1,\dots,m$, $\mathrm{e}^{\Theta}$ is a vector whose component is $\mathrm{e}^{\Theta_{i}}$, $./$ denotes the component-wise division, $I_{1}$ is the index matrix of the domain $\Omega'$ and $I_{2}$ is the index matrix of the landmarks.

\subsubsection{Subproblem $\Theta$ in \eqref{SIteraiveScheme}}\label{SP:Theta}
Now, we consider solving the subproblem $\Theta$ in \eqref{SIteraiveScheme}. Note that the subproblem $\Theta$ in \eqref{SIteraiveScheme} is nonconvex because it involves the term $|\mathrm{e}^{\Theta}-s^{k}|_{2}^{2}$. Here, we choose the Gauss-Newton method, which picks a symmetric positive definite part of its full Hessian as the approximated Hessian and then solves the resulting Gauss-Newton system to obtain a descent search direction.

We first compute the gradient and Hessian of the subproblem $\Theta$ in \eqref{SIteraiveScheme}, respectively:
\begin{equation}\label{sp1:gradandH}
\left\{
\begin{split}
d_{\Theta} &= \alpha_{1}h\Theta - \frac{\alpha_{2}h}{2} r^{k}./\mathrm{e}^{\Theta} + \alpha_{4}h I_{1}^{T}(I_{1}\Theta-\bar{\Theta})+\rho_{1} h\diag(\mathrm{e}^{\Theta})(\mathrm{e}^{\Theta}-s^{k}), \\
H_{\Theta} &= \alpha_{1}h I + \frac{\alpha_{2}h}{2}\diag(r^{k}./\mathrm{e}^{\Theta})+ \alpha_{4}h I_{1}^{T}I_{1} + \rho_{1} h \diag(2\mathrm{e}^{2\Theta}-\mathrm{e}^{\Theta}\odot s^{k}),
\end{split}
\right.
\end{equation}
where $\diag(v)$ represents a diagonal matrix whose diagonal entries are $v$. 
To get a descent search direction, we choose 
\begin{equation}\label{sp1:approxH}
\hat{H}_{\Theta} = \alpha_{1}h I + \frac{\alpha_{2}h}{2}\diag(r^{k}./\mathrm{e}^{\Theta})+ \alpha_{4}h I_{1}^{T}I_{1} +  \rho_{1} h \diag(\mathrm{e}^{2\Theta})
\end{equation}
as the approximated Hessian and the resulting Gauss-Newton system is as follows:
\begin{equation}\label{sp1:GNsystem}
\hat{H}_{\Theta}p_{\Theta} = -d_{\Theta}.
\end{equation}
Then the iterative scheme of the subproblem $\Theta$ in \eqref{SIteraiveScheme} is 
\begin{equation}
\Theta^{i+1} = \Theta^{i} + \delta^{i}_{\Theta}p^{i}_{\Theta},
\end{equation}
where $\delta^{i}_{\Theta}$ is the $i$-th step length obtained by using the Armijo line search \cite{nocedal2006numerical}. The steps are summarized in the following Algorithm \ref{Alg:SP1}.

\begin{algorithm}[h!]
\caption{Solving Subproblem $\Theta$ in \eqref{SIteraiveScheme}.}
\label{Alg:SP1}
\begin{algorithmic}
\STATE{Set $\Theta^{0}$ and $i=0$; Compute $d_{\Theta}^{i}$ and $\hat{H}_{\Theta}^{i}$ from \eqref{sp1:gradandH} and \eqref{sp1:approxH};}
\WHILE{stopping criteria is not satisfied}
\STATE{Solve $\hat{H}_{\Theta}^{i} p_{\Theta}^{i}=-d_{\Theta}^{i}$;}
\STATE{Update $\Theta^{i+1}$ by Armijo line search;}
\STATE{Set $i = i+1$;}
\STATE{Compute $d_{\Theta}^{i}$ and $\hat{H}_{\Theta}^{i}$ from \eqref{sp1:gradandH} and \eqref{sp1:approxH};}
\ENDWHILE
\end{algorithmic}
\end{algorithm}

\subsubsection{Subproblem $Y$ in \eqref{SIteraiveScheme}}\label{SP:Y}
Here, for the subproblem $Y$ in \eqref{SIteraiveScheme}, we again use the Gauss-Newton method. Although the cost of computing the approximated Hessian and solving a linear system may be expensive, in practice we find that it is still much faster than first-order methods such as L-BFGS\cite{liu1989limited}. This is because first-order methods usually require much more iterations and computing interpolation many times, which can consume a lot of running time.

To implement the Gauss-Newton method, we need to compute the gradient and the approximated Hessian of subproblem $Y$ in \eqref{SIteraiveScheme}, respectively:
\begin{equation}\label{sp2:gradandH}
\begin{split}
d_{Y} &=  \alpha_{2}hM_{1}^{T}(1./\mathrm{e}^{\Theta^{k+1}}) + \alpha_{3}hB^{T}BY +\alpha_{5}hP^{T}\vec T_{PY}^{T}(\vec T(PY)-\vec R) \\
& + \rho_{1} h M_{2}^{T}(A_{1}Y\odot A_{4}Y - A_{2}Y\odot A_{3}Y - \mathrm{e}^{\Theta^{k+1}} +\frac{\Lambda_{1}^{k}}{\rho_{1}}) + \rho_{2} I_{2}^{T}(I_{2}Y-Q+\frac{\Lambda_{2}}{\rho_{2}}), \\
\hat{H}_{Y} &=  \alpha_{2}hM_{1}^{T}M_{1}+\alpha_{3}hB^{T}B+\alpha_{5}hP^{T}\vec T_{PY}^{T}\vec T_{PY}P + \rho_{1} h M_{2}^{T}M_{2} + \rho_{2} I_{2}^{T}I_{2},
\end{split}
\end{equation}
where $\vec T_{PY}$ is the Jacobian of $\vec T$ with respect to $PY$, $M_{1} = \sum_{l=1}^{4}\diag(A_{l}Y)A_{l}$ and $M_{2} = \diag(A_{1}Y)A_{4}+\diag(A_{4}Y)A_{1}-\diag(A_{2}Y)A_{3}-\diag(A_{3}Y)A_2$.

Hence, the iterative scheme of the subproblem $Y$ in \eqref{SIteraiveScheme} is 
\begin{equation}
Y^{i+1} = Y^{i} + \delta^{i}_{Y}p^{i}_{Y},
\end{equation}
where $p^{i}_{Y}$ is solved by the Gauss-Newton system $\hat{H}^{i}_{Y}p^{i}_{Y} = - d^{i}_{Y}$ and $\delta^{i}_{Y}$ is the $i$-th step length obtained again by using the Armijo line search. The steps are summarized in the following Algorithm \ref{Alg:SP2}.

\begin{algorithm}[h!]
\caption{Solving Subproblem $Y$ in \eqref{SIteraiveScheme}.}
\label{Alg:SP2}
\begin{algorithmic}
\STATE{Set $Y^{0}$ and $i=0$; Compute $d_{Y}^{i}$ and $\hat{H}_{Y}^{i}$ from \eqref{sp2:gradandH};}
\WHILE{stopping criteria is not satisfied}
\STATE{Solve $\hat{H}_{Y}^{i} p _{Y}^{i}=-d_{Y}^{i}$;}
\STATE{Update $Y^{i+1}$ by Armijo line search;}
\STATE{Set $i = i+1$;}
\STATE{Compute $d_{Y}^{i}$ and $\hat{H}_{Y}^{i}$ from \eqref{sp2:gradandH};}
\ENDWHILE
\end{algorithmic}
\end{algorithm}

\begin{remark}
The Gauss-Newton method has been successfully used to solve the resulting optimization problem derived by different image registration models \cite{burger2013hyperelastic,zhang2018novel}. Note that in \cite{burger2013hyperelastic,zhang2018novel}, the line search strategy has to satisfy the sufficient descent condition and simultaneously guarantee that the transformation is diffeomorphic under the sense of discretization. However, in this paper, due to the nonnegative constraint, we only need to focus on the sufficient descent condition and the diffeomorphic property is ensured by the quadratic penalty term in the Lagrangian function \eqref{ALF}.
\end{remark}

\subsubsection{Summary of ADMM}
In Section \ref{SP:Theta} and \ref{SP:Y}, we have discussed how to use the Gauss-Newton method to solve the subproblems in the iterative scheme of ADMM \eqref{SIteraiveScheme}. Altogether, the procedure of solving the proposed model \eqref{PM} is summarized in Algorithm \ref{Alg:ADMM}. Note that in the standard ADMM, the penalty parameters $\rho_{1}$ and $\rho_{2}$ are fixed. Changing these penalty parameters $\rho_{1}$ and $\rho_{2}$ dynamically may heavily affect the convergence rate of ADMM \cite{boyd2011distributed}. In practice, we set $E = A_{1}Y\odot A_{4}Y - A_{2}Y\odot A_{3}Y-\mathrm{e}^{\Theta}$ and enlarge $\rho_{1}$ by a factor of $2$ if $|E^{k+1}|_{\infty} > 0.95|E^{k}|_{\infty}$ and keep $\rho_{2}$ fixed.

\begin{algorithm}[h!]
\caption{Solving the proposed model \eqref{PM} by ADMM.}
\label{Alg:ADMM}
\begin{algorithmic}
\STATE{Input $T, R$, volume prior $\bar\Theta$ and Landmarks $(p_{i},q_{i}), i=1,\dots,m$. Set $Y^{0},\Theta^{0}, \Lambda_{1}^{0}, \Lambda_{2}^{0}, \rho_{1}^{0}, \rho_{2}^{0}$, and $k=1$; 	Give parameters $\alpha_{i}, i=1,\dots,5$;} 
\WHILE{stopping criteria is not satisfied}
\STATE{Compute $\Theta^{k+1}$ by Algorithm \ref{Alg:SP1};}
\STATE{Compute $Y^{k+1}$ by Algorithm \ref{Alg:SP2};}
\STATE{Compute $\Lambda_{1}^{k+1}$ by $\Lambda_{1}^{k} + \rho_{1}^{k}(A_{1}Y^{k+1}\odot A_{4}Y^{k+1} - A_{2}Y^{k+1}\odot A_{3}Y^{k+1}-\mathrm{e}^{\Theta^{k+1}})$;}
\STATE{Compute $\Lambda_{2}^{k+1}$ by $\Lambda_{2}^{k} + \rho_{2}^{k}(I_{2}Y^{k+1}-Q)$;}
\IF{$|E^{k+1}|_{\infty} > 0.95|E^{k}|_{\infty}$}
\STATE{Update $\rho^{k+1} = 2\rho^{k}$; }
\ENDIF
\STATE{Set $k = k +1$;}
\ENDWHILE
\end{algorithmic}
\end{algorithm}

Note that there are some existing convergence results of ADMM in nonconvex optimization~\cite{wang2019global} but they are not applicable to the proposed Algorithm \ref{Alg:ADMM} as the equality constraint in \eqref{PM} is nonlinear. Although the convergence of Algorithm \ref{Alg:ADMM} cannot be guaranteed in theory, we can show that Algorithm \ref{Alg:SP1} and Algorithm \ref{Alg:SP2}, which solve the subproblems in \eqref{SIteraiveScheme}, are convergent. To prove this result, we first recall the following lemma.

\begin{lemma}[see~\cite{kelley1999iterative}]\label{Lemma}
Consider the following unconstrained optimization problem:
\begin{equation}
\min_{x} f(x),
\end{equation}
where $f:\mathbb{R}^{n}\rightarrow\mathbb{R}$ is a twice differentiable function. Let $\nabla f$ be Lipschitz continuous with Lipschitz constant $L$. Assume that the approximated Hessian matrices $\hat{H}^{i}$ are symmetric positive definite and that there are $\xi_{1}$ and $\xi_{2}$ such that $\kappa(\hat{H}^{i})\leq \xi_{1}$, and $\|\hat{H}^{i}\|\leq \xi_{2}$ for all $i$. Then either $f(x^{i})$ is unbounded from below or 
\begin{equation}
\lim_{i\rightarrow\infty} \nabla f(x^{i}) = 0
\end{equation}
and hence any limit point of the sequence of iterates produced by the Gauss-Newton method with Armijo line search is a stationary point.

In particular, if $f(x^{i})$ is bounded from below and $x_{i_{l}}\rightarrow x^{*}$ is any convergent subsequence of $\{x^{i}\}$, then $\nabla f(x^{*}) = 0$.
\end{lemma}

Based on Lemma \ref{Lemma}, we have the following convergence result about the subproblems in \eqref{SIteraiveScheme}.
\begin{theorem}\label{Theorem}
Assume that $T$ and $R$ are twice differentiable. Let the parameters $\alpha_{l},l=1,\dots,5, \rho$ and $\sigma$ be positive. In addition, assume that $\alpha_{3}B^{T}B+\rho_{2} I_{2}^{T}I_{2}$ is symmetric positive definite. Then Algorithm \ref{Alg:SP1} and Algorithm \ref{Alg:SP2}, which solve the subproblems in \eqref{SIteraiveScheme}, are globally convergent. 
\end{theorem}

\begin{proof}
For Algorithm \ref{Alg:SP1}, from \eqref{sp1:approxH}, since each component of $r^{k}$ is nonnegative, it is obvious that $\hat{H}^{i}_{\Theta}$ is symmetric positive definite. In addition, it is also clear that $J_{1}(\Theta)$ tends to infinity when some component of $\Theta$ goes to infinity. Hence, the sublevel set $\{\Theta | J_{1}(\Theta)\leq J_{1}(\Theta^{0})\}$ can be covered by a compact (closed and bounded) set $\mathcal{U}_{1}$. Then $\nabla J_{1}(\Theta)$ is Lipschitz continuous with Lipschitz constant $L_{1}$, and $\kappa(\hat{H}_{\Theta}^{i})$ and $\|\hat{H}_{\Theta}^{i}\|$ for all $i$ have the upper bound under the set $\mathcal{U}_{1}$, respectively. In addition, $0$ is a lower bound of $J_{1}(\Theta)$. Therefore, by Lemma \ref{Lemma}, we have 
\begin{equation}
\lim_{i\rightarrow\infty} \nabla J_{1}(\Theta^{i}) = 0.
\end{equation}

For Algorithm \ref{Alg:SP2}, since $\alpha_{3}B^{T}B+\sigma I_{2}^{T}I_{2}$ is assumed to be symmetric positive definite, $\hat{H}^{i}_{Y}$ for all $i$ in \eqref{sp2:gradandH} is symmetric positive definite. We also have that $J_{2}(Y)$ tends to infinity when some component of $Y$ goes to infinity. This is because $A_{1}Y\odot A_{4}Y - A_{2}Y\odot A_{3}Y$ represents the discrete Jacobian determinant and the ratio of change of the area. When $Y$ goes to infinity, the area will expand to infinity and then $A_{1}Y\odot A_{4}Y - A_{2}Y\odot A_{3}Y$ will go to infinity. Hence, a compact set $\mathcal{U}_{2}$ can cover the sublevel set $\{Y | J_{2}(Y)\leq J_{2}(Y^{0})\}$. Following a similar discussion in the first part, we can conclude that
\begin{equation}
\lim_{i\rightarrow\infty} \nabla J_{2}(Y^{i}) = 0.
\end{equation}

\end{proof}

\begin{remark}
In Theorem \ref{Theorem}, we assume that $\alpha_{3}B^{T}B^{B}+\rho_{2} I_{2}^{T}I_{2}$ is symmetric positive definite, which can be guaranteed when the Dirichlet boundary condition is employed or the Neumann boundary condition is employed and $\rho_{2}$ is positive. In addition, we also let the parameters $\alpha_{l},l=1,\dots,5, \rho_{1}$ and $\rho_{2}$ be positive. In fact, this requirement can be relaxed to: (1) $\rho_{1} = 0$ with no landmarks; (2) $\alpha_{2}=0$ with no fitting term; (3) $\alpha_{4}=0$ with no volume prior. Theorem \ref{Theorem} can still be established under these conditions.
\end{remark}
 
\begin{remark}
While Theorem~\ref{Theorem} guarantees the global convergence of Algorithm~\ref{Alg:SP1} and Algorithm~\ref{Alg:SP2} for solving the subproblems in \eqref{SIteraiveScheme}, the global convergence of the entire ADMM framework (Algorithm~\ref{Alg:ADMM}) has not been proved due to the nonlinear equality constraint \eqref{PM}. Nevertheless, we find that the solutions satisfy the KKT conditions in all our experiments. As the solutions may not be unique, from this point of view, the current approach only finds a stationary point which may not necessary be the global optimum. 
\end{remark}

\section{Experiments}\label{sect:experiment}
In this section, we test the proposed framework \eqref{PM} with several synthetic examples. The algorithms are implemented in MATLAB R2019a and tested on a MacBook Pro with 2.2 GHz Quad-Core Intel Core i7 processor and 16 GB RAM.

\begin{figure}[t!]
\centering
\subfigure[Landmarks]{
\includegraphics[width=1.8in,height=1.8in]{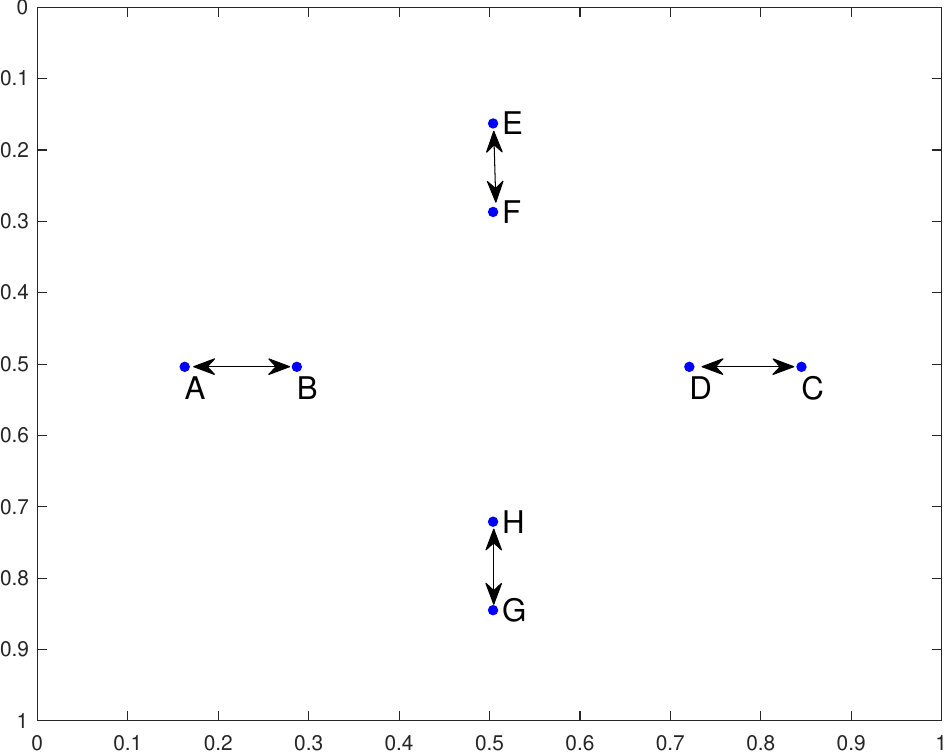}}
\subfigure[Transformation $\bm{y}$]{
\includegraphics[width=1.8in,height=1.8in]{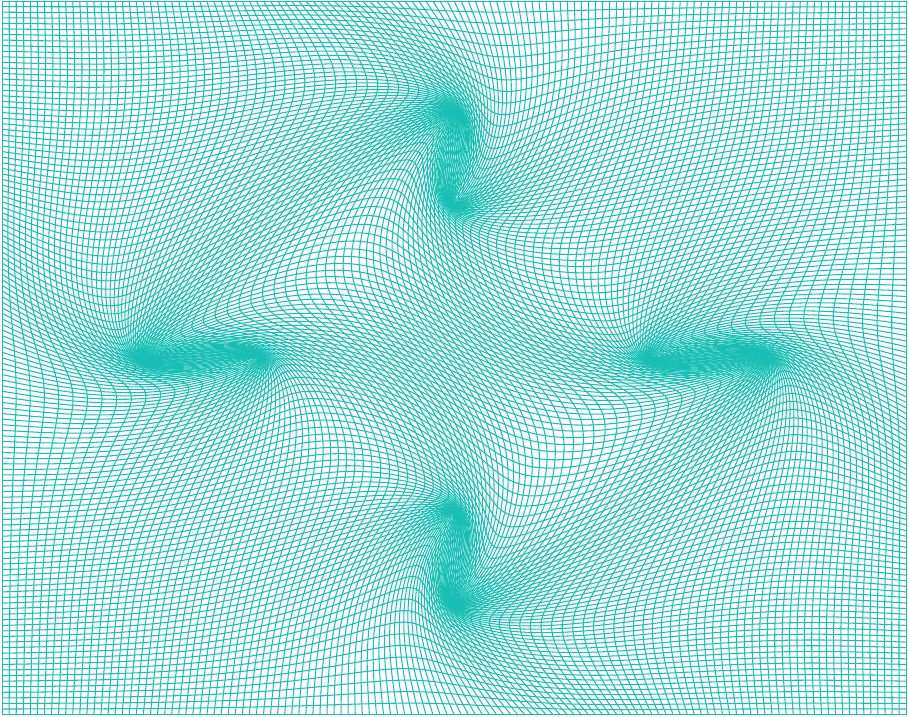}}
\subfigure[$\det\nabla\bm{y}$]{
\includegraphics[width=1.8in,height=1.8in]{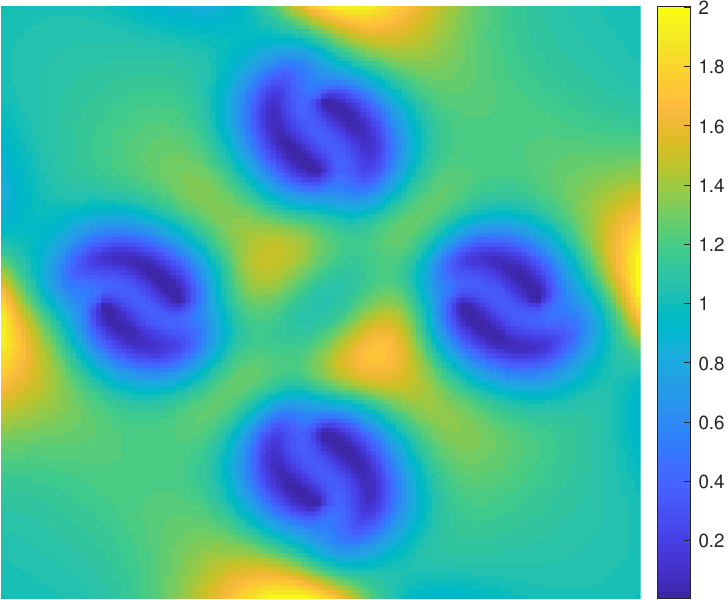}}
\caption{An example of landmark-constrained 2D quasi-conformal mapping achieved by our proposed framework. (a) The corresponding landmarks. (b)--(c) The transformation and the corresponding Jacobian determinant.}\label{fig:landmark_2D}
\end{figure}

\subsection{Landmark-constrained $n$-dimensional quasi-conformal mapping}
We first test our proposed model for landmark-constrained $n$-dimensional quasi-conformal mappings as described in~\eqref{case1}. Here, we set $\alpha_{2}=1$ and $\alpha_{3}=0.01$.

Consider a 2D example with 8 prescribed landmarks $A, B, C, D, E, F, G, H$ in the domain $[0,1]\times [0,1]$ (see Fig.~\ref{fig:landmark_2D}(a)). The landmarks are divided into 4 pairs $(A,B)$, $(C,D)$, $(E,F)$, $(G,H)$, and for each pair of landmarks we constrain the points to swap their positions under the quasi-conformal mapping, which involves a large deformation and could easily lead to overlaps. Using our proposed model, we can easily achieve a smooth transformation $\bm{y}$ that satisfies the landmark constraints (see Fig.~\ref{fig:landmark_2D}(b)). By computing the Jacobian determinant $\det\nabla\bm{y}$ as shown in Fig.~\ref{fig:landmark_2D}(c), we can see that $\det\nabla\bm{y}$ is always greater than 0 and hence the mapping is folding-free.

We then consider a 3D example with 8 prescribed landmarks $A, B, C, D, E, F, G, H$ in the domain $[0,1]\times [0,1]\times [0,1]$ (see Fig.~\ref{fig:landmark_3D}(a)). This time, the landmarks are divided into two groups $(A,B,C,D)$ and $(E,F,G,H)$, and for each group of landmarks we enforce the points to their neighbouring points to form two large-scale twists at two different angles. As shown in Fig.~\ref{fig:landmark_3D}(b)--(f), our model is capable of producing a smooth and folding-free 3D mapping that satisfies the landmark constraints (see also Supplementary Movie 1).

\begin{figure}[t!]
\centering
\includegraphics[width=\textwidth]{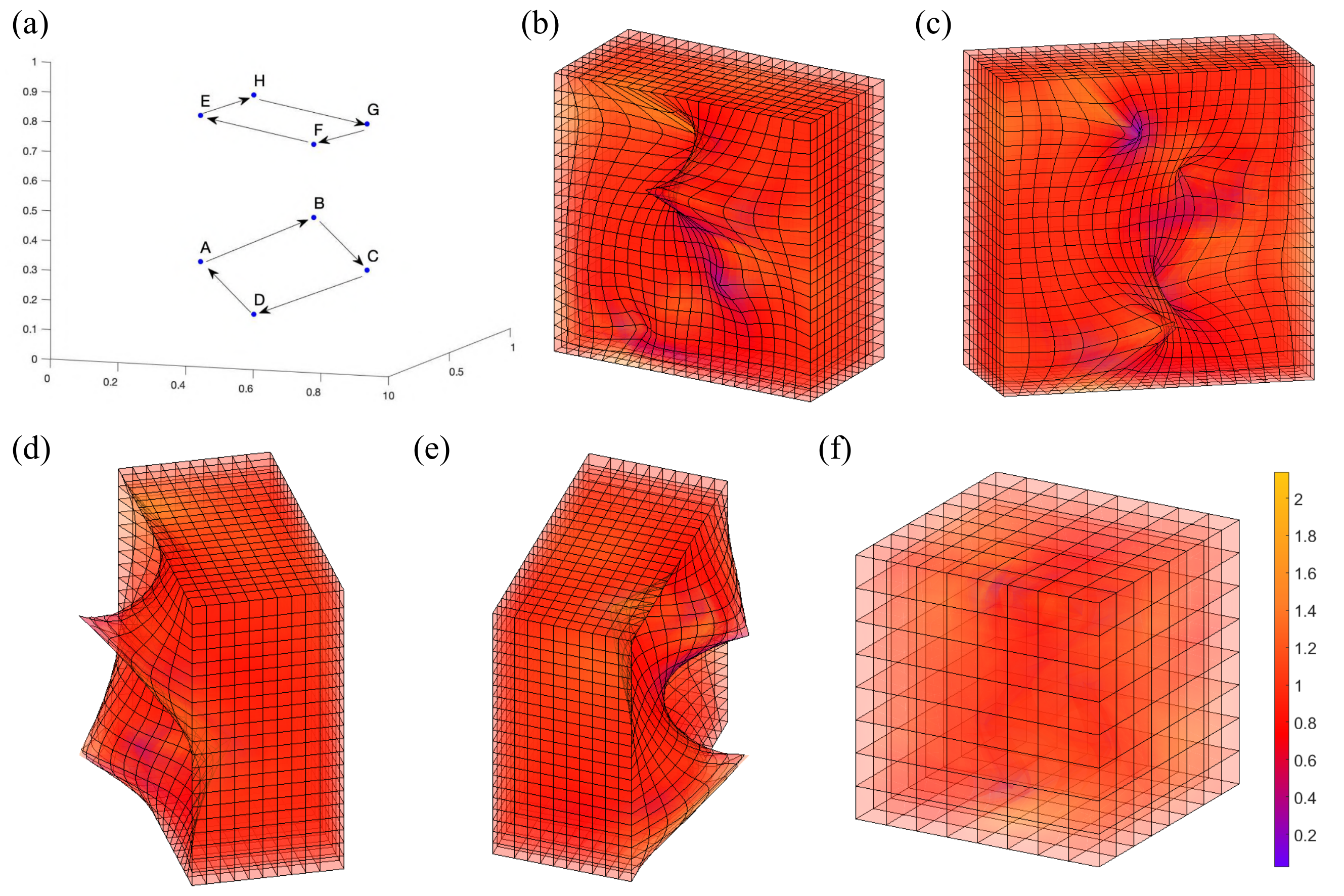}
\caption{An example of landmark-constrained 3D quasi-conformal mapping achieved by our proposed framework. (a) The 8 landmarks with their target positions indicated by the arrows. (b) A cross-sectional view of the transformation $\bm{y}$. (c)--(e) Several rotated versions of (b) for visualizing the large 3D deformation (see also Supplementary Movie 1). (f) The corresponding Jacobian determinant $\det\nabla\bm{y}$.}\label{fig:landmark_3D}
\end{figure}

\begin{figure}[t!]
\centering
\subfigure[Energy by the proposed algorithm.]{
\includegraphics[width=2.5in,height=2.0in]{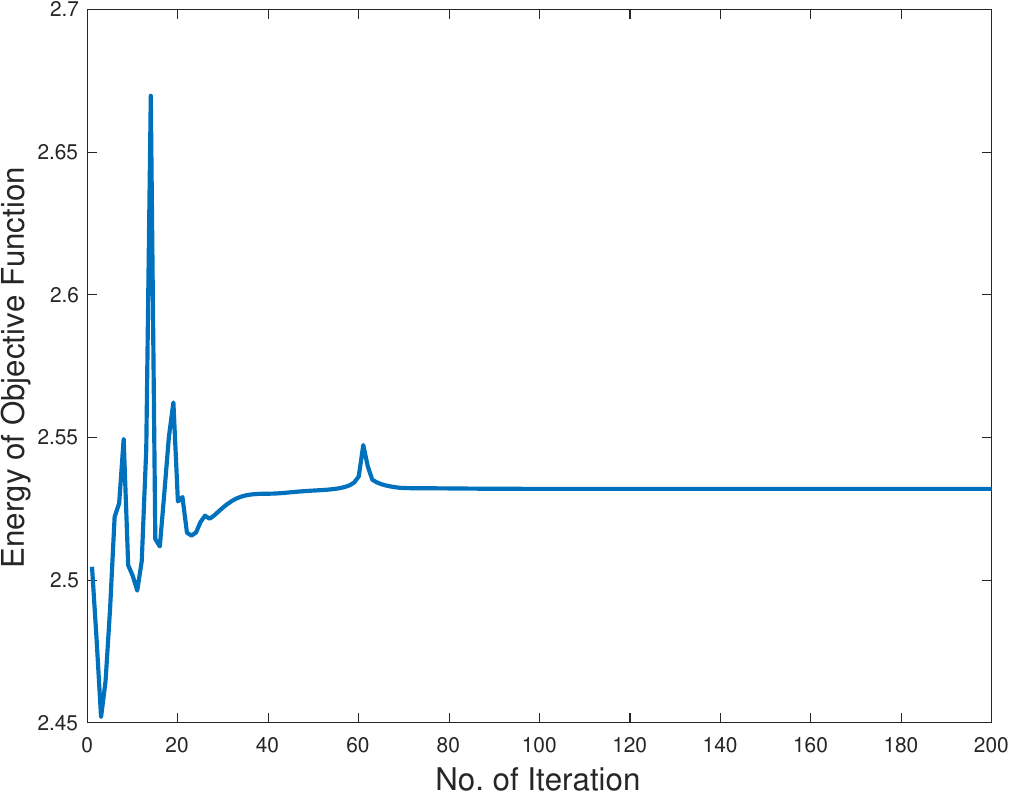}}
\subfigure[Energy by \cite{lee2016landmark}]{
\includegraphics[width=2.5in,height=2.0in]{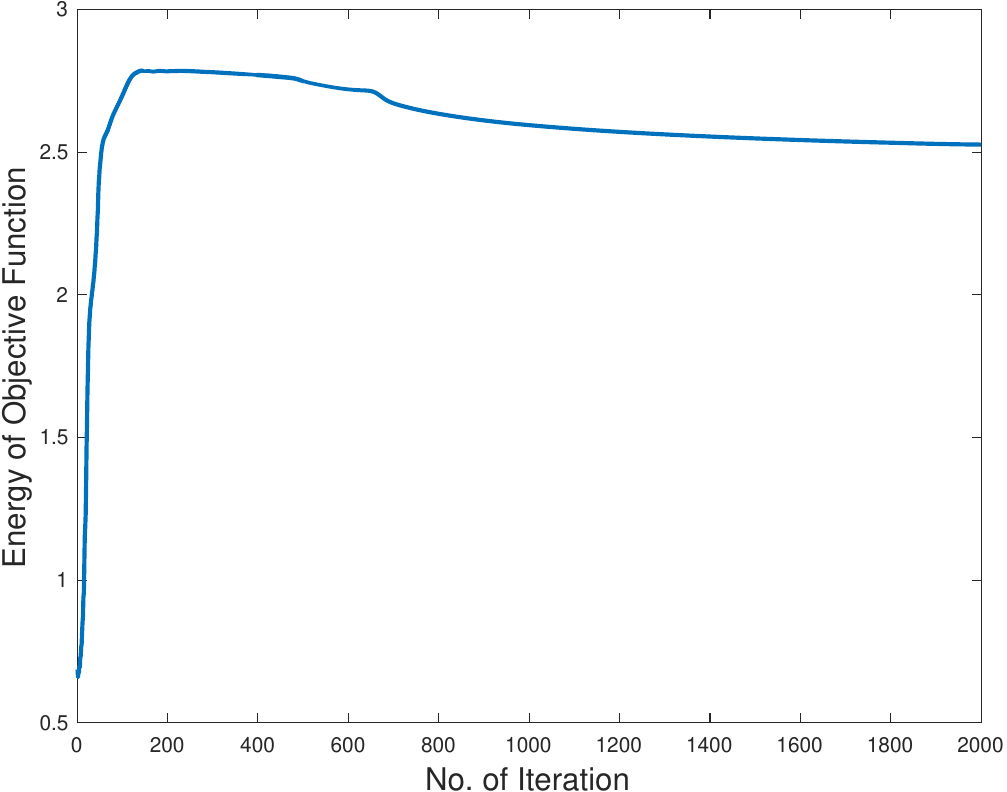}}\\
\subfigure[Violation of equality constraint by the proposed algorithm]{
\includegraphics[width=2.5in,height=2.0in]{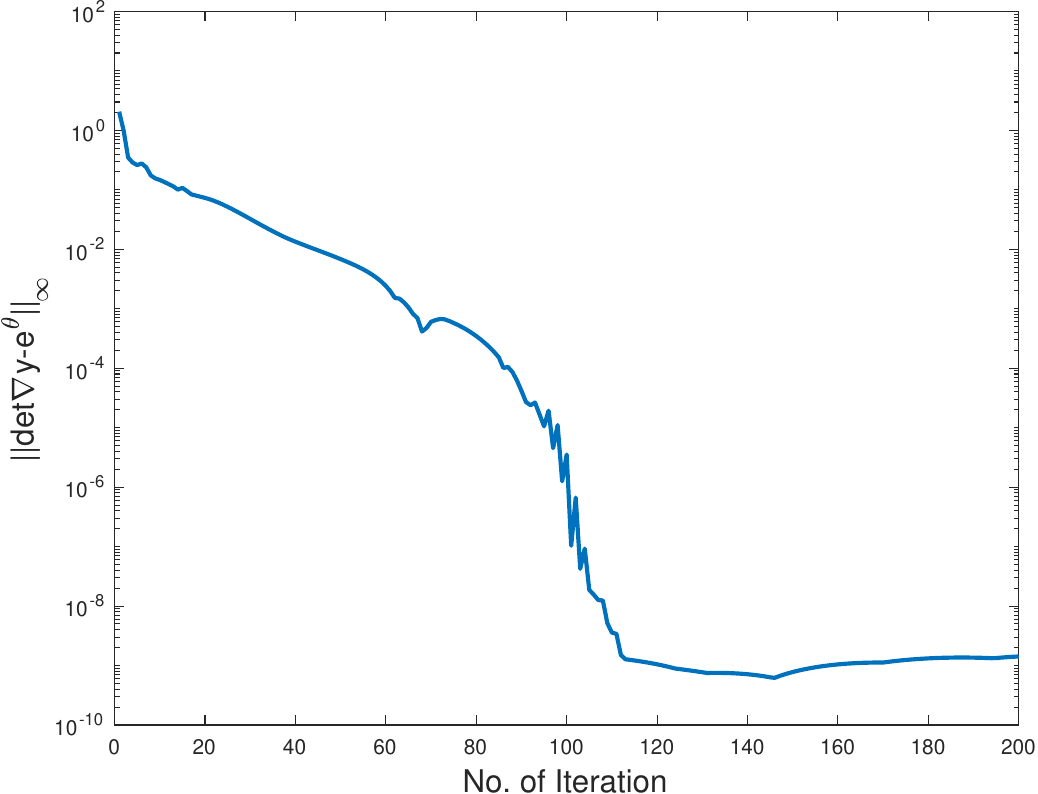}}
\subfigure[Violation of equality constraint by \cite{lee2016landmark}]{
\includegraphics[width=2.5in,height=2.0in]{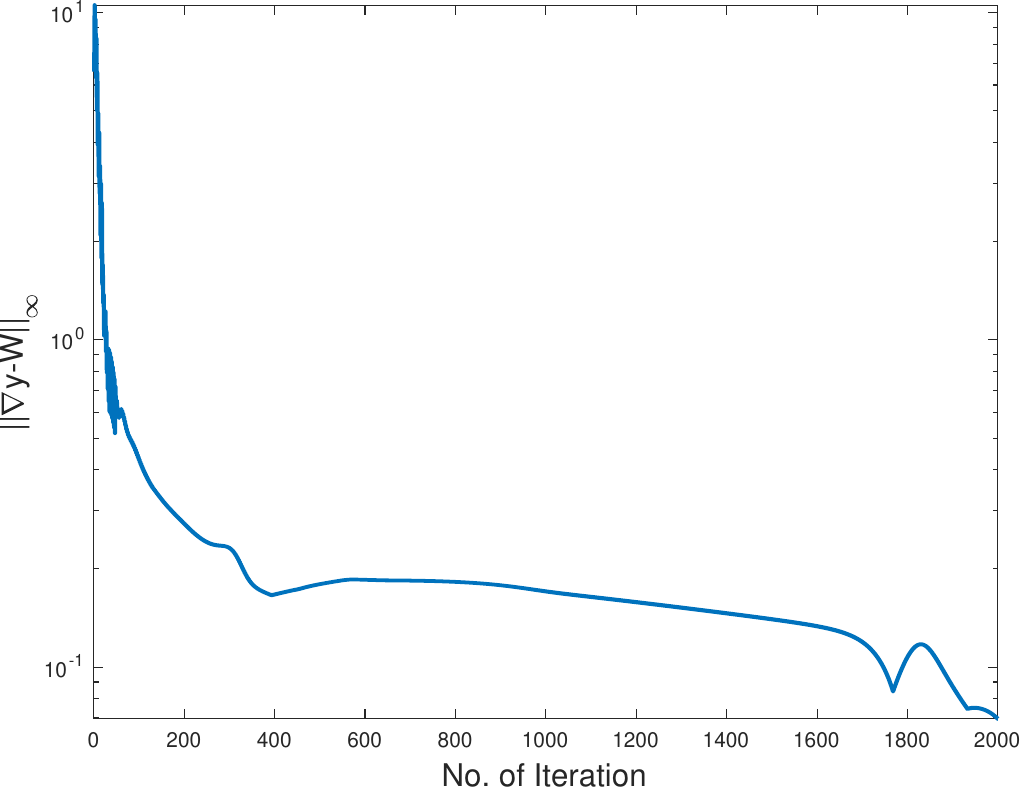}}
\caption{Comparison between the proposed algorithm and the algorithm in~\cite{lee2016landmark} for landmark-constrained $n$-dimensional quasi-conformal mapping. The first row displays the energy generated by the proposed algorithm (Algorithm~\ref{Alg:ADMM}) and the algorithm in \cite{lee2016landmark}. The second row gives the gaps of the equality constraint by these two algorithms. It is clear that these two algorithms achieve a similar final energy but the proposed algorithm converges much faster. More specifically, the proposed algorithm only takes about 100 iterations to reach $10^{-8}$ for the violation of the equality constraint with respect to the infinity norm, while the prior method in~\cite{lee2016landmark} requires over 2000 iterations and only reaches $10^{-1}$ for the constraint violation.}\label{fig:energy}
\end{figure}

It is natural to compare the performance of our model with the existing landmark-constrained $n$-dimensional quasi-conformal mapping method in~\cite{lee2016landmark}. For the algorithm in~\cite{lee2016landmark}, to apply ADMM, it introduces an auxiliary variable $\bm{v}$ such that $\bm{v}=\nabla\bm{y}$ and then substitutes $\bm{v}$ into the denominator of $K$ in \eqref{nDistortion}. A comparison between the performance of the proposed algorithm and the algorithm in~\cite{lee2016landmark} is provided in Fig.~\ref{fig:energy}, from which it can be observed that our proposed model outperforms the method~\cite{lee2016landmark}. More specifically, note that both algorithms are capable of achieving a similar final energy. However, for the violation of the equality constraint, the proposed algorithm only needs about 100 iterations to reach $10^{-8}$ with respect to the infinity norm while the method in~\cite{lee2016landmark} requires 2000 iterations and only reaches $10^{-1}$ for the constraint violation. In other words, the proposed method is significantly faster and hence more practical when compared to the prior method~\cite{lee2016landmark}. The main reason for the improvement achieved by our method is the use of the exponential term $\mathrm{e}^{\theta}$ for the Jacobian determinant of the transformation, which largely simplifies the computation for the energy minimization problem.

\subsection{Landmark- and intensity-based $n$-dimensional quasi-conformal registration}
Next, we test our proposed model for landmark- and intensity-based $n$-dimensional quasi-conformal registration as described in~\eqref{case2}. Here, we set $\alpha_{2}=1$, $\alpha_{3}=0.01$ and $\alpha_{5} = 10^4$.

Again, we start with a 2D example as shown in Fig.~\ref{fig:intensity_landmark_2D}. In this example, the goal is to obtain a registration between an ``I'' shape and a ``C'' shape (see Fig.~\ref{fig:intensity_landmark_2D}(a)--(b)). In addition to the intensity of the two letters, six landmarks are prescribed to represent the correspondence between the two shapes. This makes the proposed model~\eqref{case2} well-suited for computing the registration. As shown in Fig.~\ref{fig:intensity_landmark_2D}(c)--(d), the landmark- and intensity-based quasi-conformal registration model gives a smooth and accurate registration result with the overall shape as well as the landmarks well-matched. From Fig.~\ref{fig:intensity_landmark_2D}(e), we see that the registration is bijective. For comparison, we consider a purely intensity-based registration without any prescribed landmarks (Fig.~\ref{fig:intensity_landmark_2D}(f)--(h)). It can be observed that while the overall intensity is matched in the registration result, the corners of the ``C'' shape is not perfectly matched. 

\begin{figure}[t!]
\centering
\subfigure[Template with landmarks]{
\includegraphics[width=1.8in,height=1.8in]{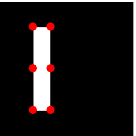}}
\subfigure[Reference with landmarks]{
\includegraphics[width=1.8in,height=1.8in]{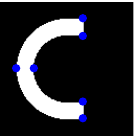}} \\
\subfigure[$T(\bm{y})$ with landmarks]{
\includegraphics[width=1.8in,height=1.8in]{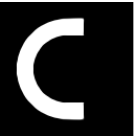}}
\subfigure[$\bm{y}$ with landmarks]{
\includegraphics[width=1.8in,height=1.8in]{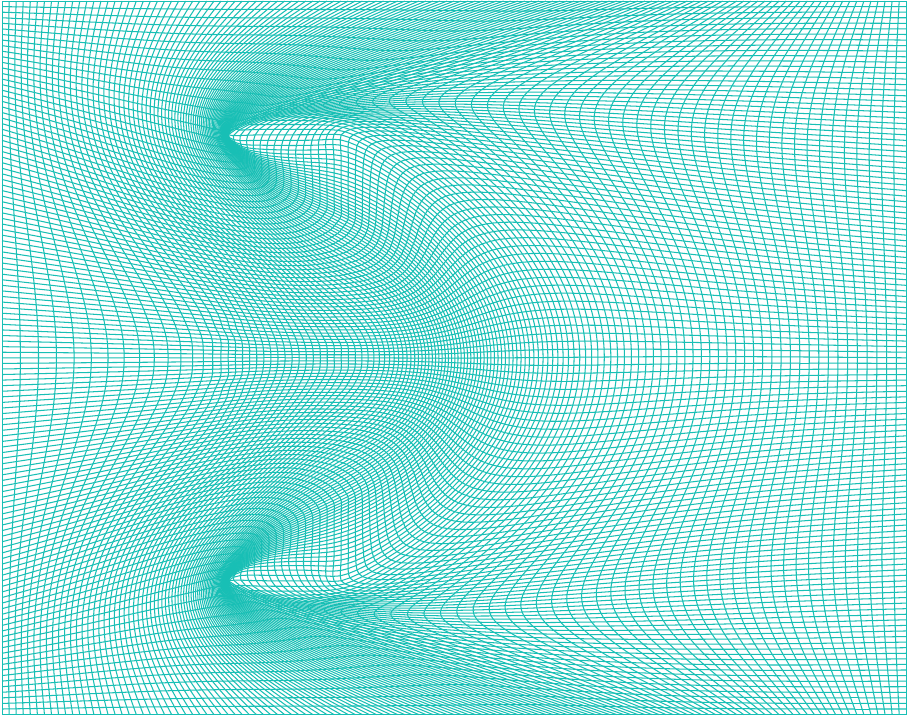}}
\subfigure[$\det\nabla\bm{y}$ with landmarks]{
\includegraphics[width=1.8in,height=1.8in]{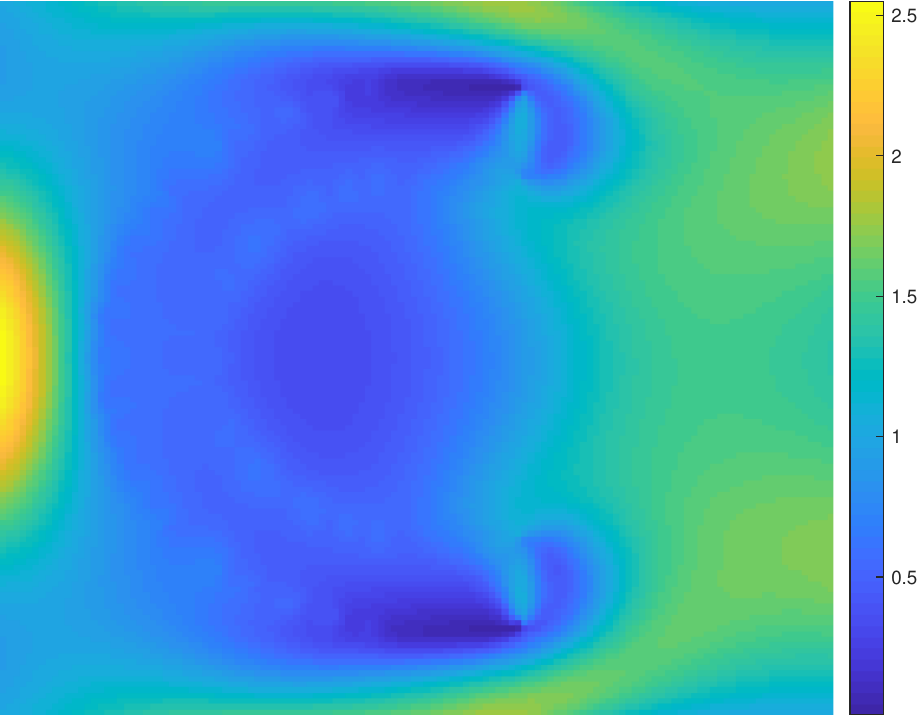}} \\
\subfigure[$T(\bm{y})$ without landmarks]{
\includegraphics[width=1.8in,height=1.8in]{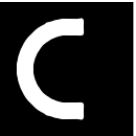}}
\subfigure[$\bm{y}$ without landmarks]{
\includegraphics[width=1.8in,height=1.8in]{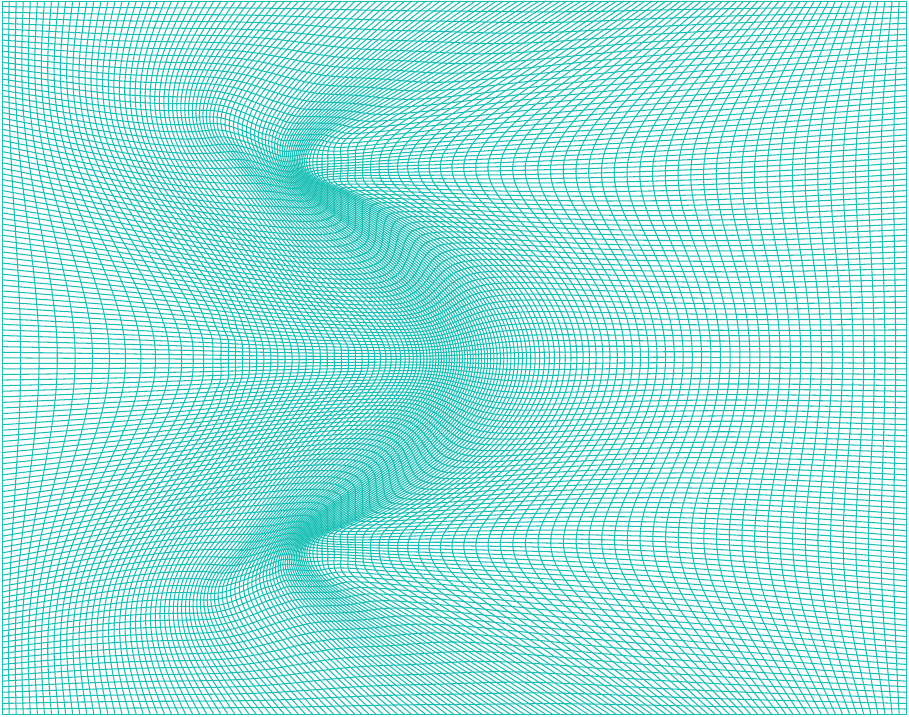}}
\subfigure[$\det\nabla\bm{y}$ without landmarks]{
\includegraphics[width=1.8in,height=1.8in]{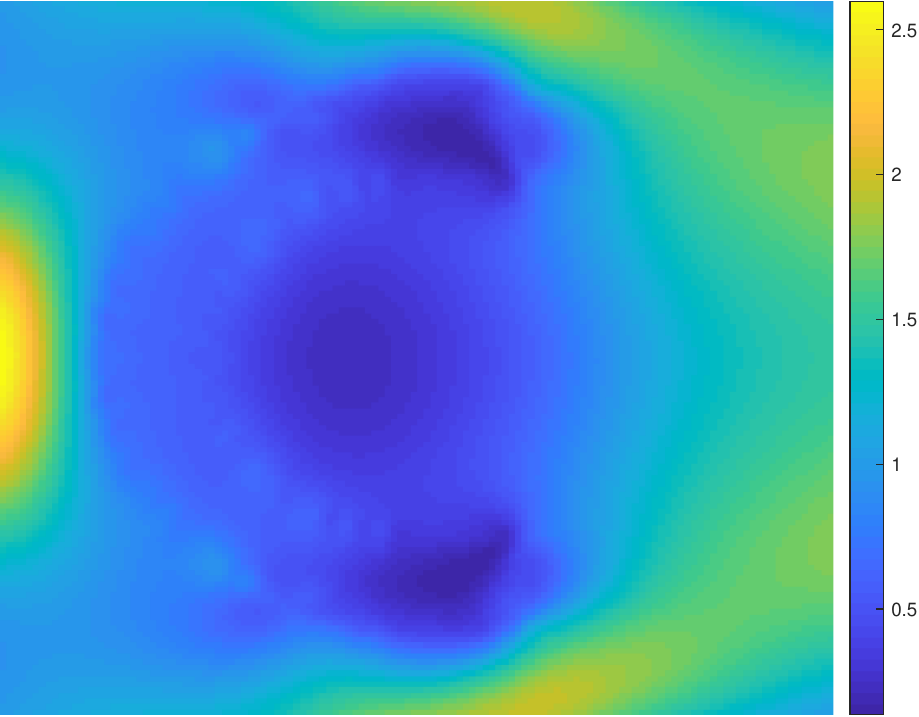}}
\caption{An example of landmark- and intensity-based 2D quasi-conformal registration achieved by our proposed framework. The first row gives the template and reference with landmarks. The second and third rows shows the deformed template, transformation and the corresponding Jacobian determinant  with and without landmarks respectively.}\label{fig:intensity_landmark_2D}
\end{figure}

\begin{figure}[t!]
\centering
\includegraphics[width=\textwidth]{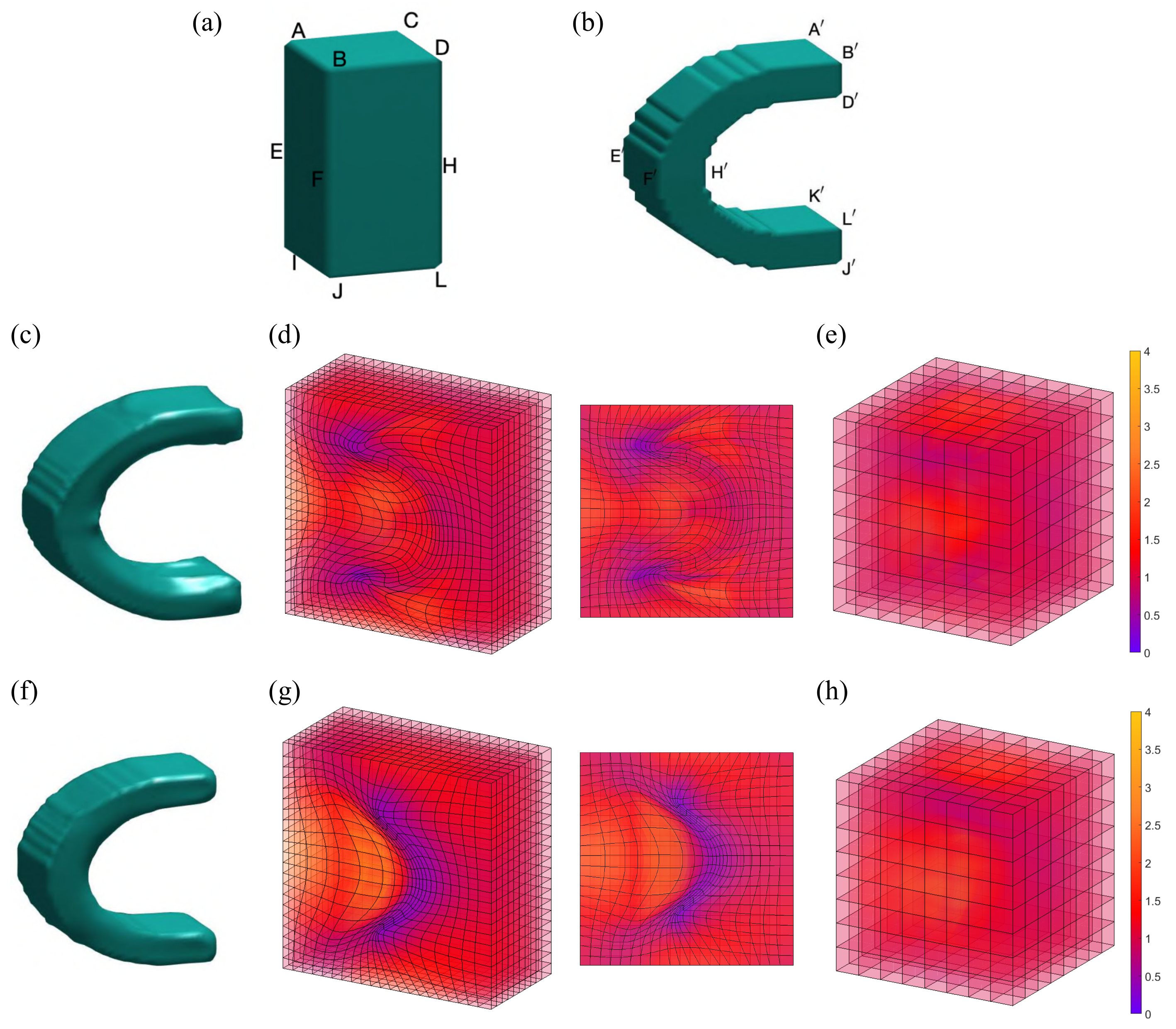}
\caption{An example of landmark- and intensity-based 3D quasi-conformal registration achieved by our proposed framework. (a) The template with landmarks. (b) The reference with landmarks. (c)--(e) The deformed template, a cross-sectional view of transformation $T(\bm{y})$, and the corresponding Jacobian determinant $\det\nabla\bm{y}$ obtained using the landmark- and intensity-based model. (f)--(h) The result obtained using a purely intensity-based model without landmark constraints.}\label{fig:intensity_landmark_3D}
\end{figure}

We then consider a 3D analog of the ``I'' to ``C'' problem as shown in Fig.~\ref{fig:intensity_landmark_3D}. This time, the template shape is a rectangular solid that represents an ``I'' shape, and the reference shape is a solid ``C'' shape. Similar to the 2D case, 12 landmarks are prescribed to represent the correspondence between the two shapes. Using the proposed model, we can compute a landmark- and intensity-based 3D quasi-conformal registration between the two shapes (Fig.~\ref{fig:intensity_landmark_3D}(c)--(e)). Even for this large deformation problem, our model is capable of producing a smooth and accurate registration result with the overall shape as well as the landmarks well-matched, and with bijectivity ensured. On the contrary, if we use a purely intensity-based model, the mapping will not be able to match the endpoints of the ``C'' shape perfectly (see Fig.~\ref{fig:intensity_landmark_3D}(f)--(h)). This demonstrates the strength of the proposed model in combining landmarks and intensity for producing an accurate registration.

\begin{figure}[t!]
\centering
\subfigure[]{
\includegraphics[width=1.4in,height=1.4in]{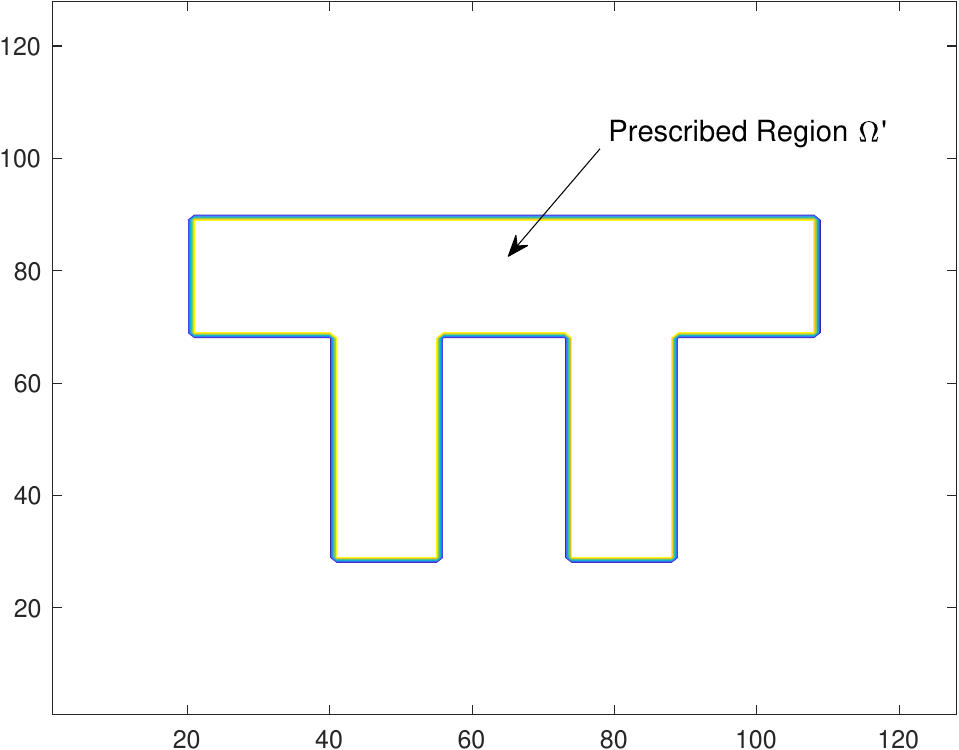}
}
\subfigure[]{
\includegraphics[width=1.4in,height=1.4in]{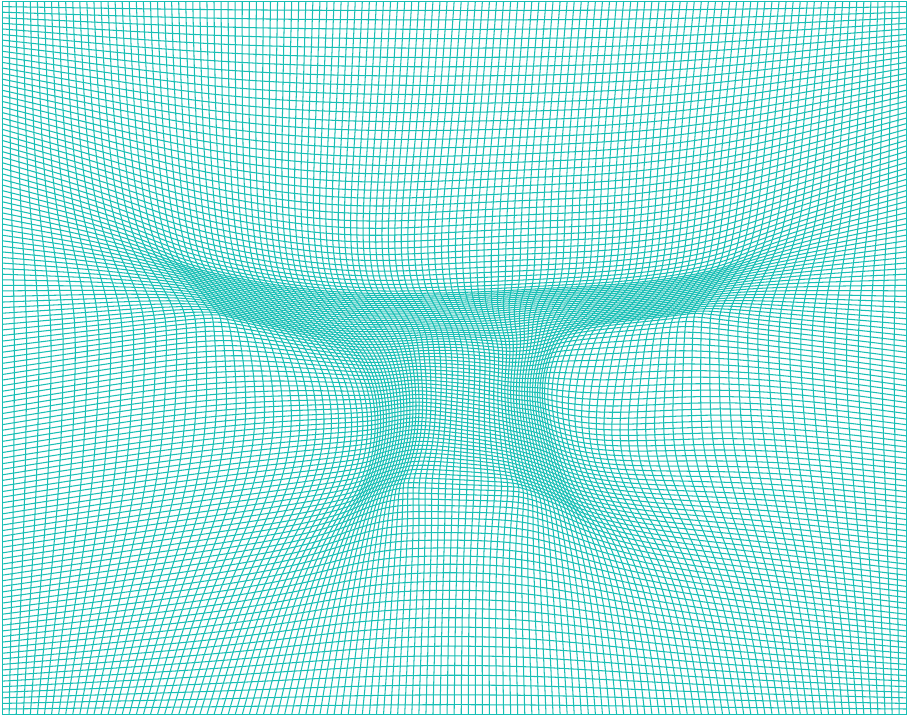}
\includegraphics[width=1.4in,height=1.4in]{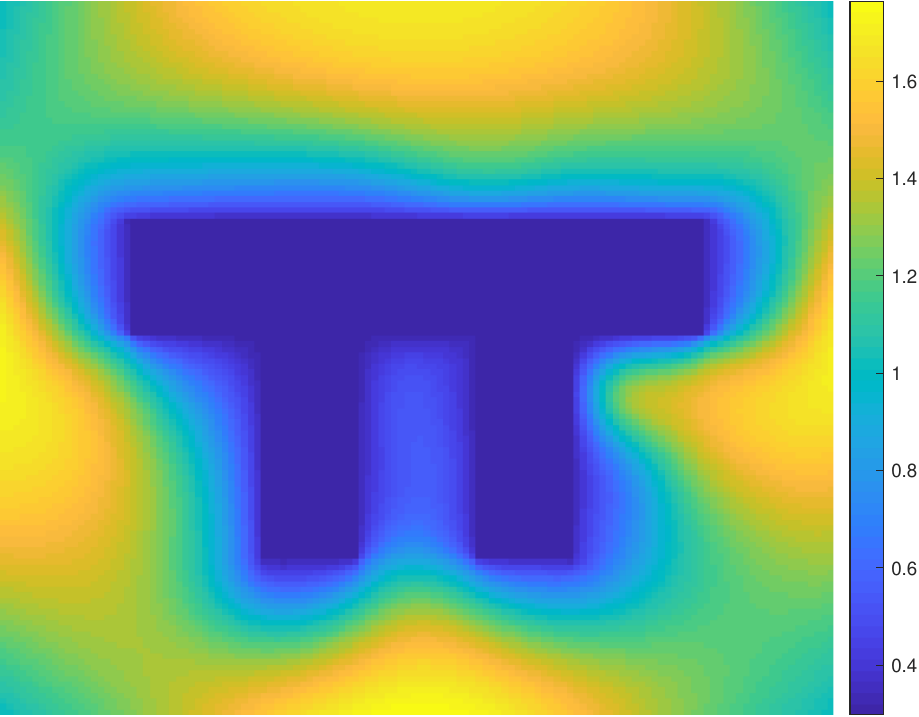}
\includegraphics[width=1.4in,height=1.4in]{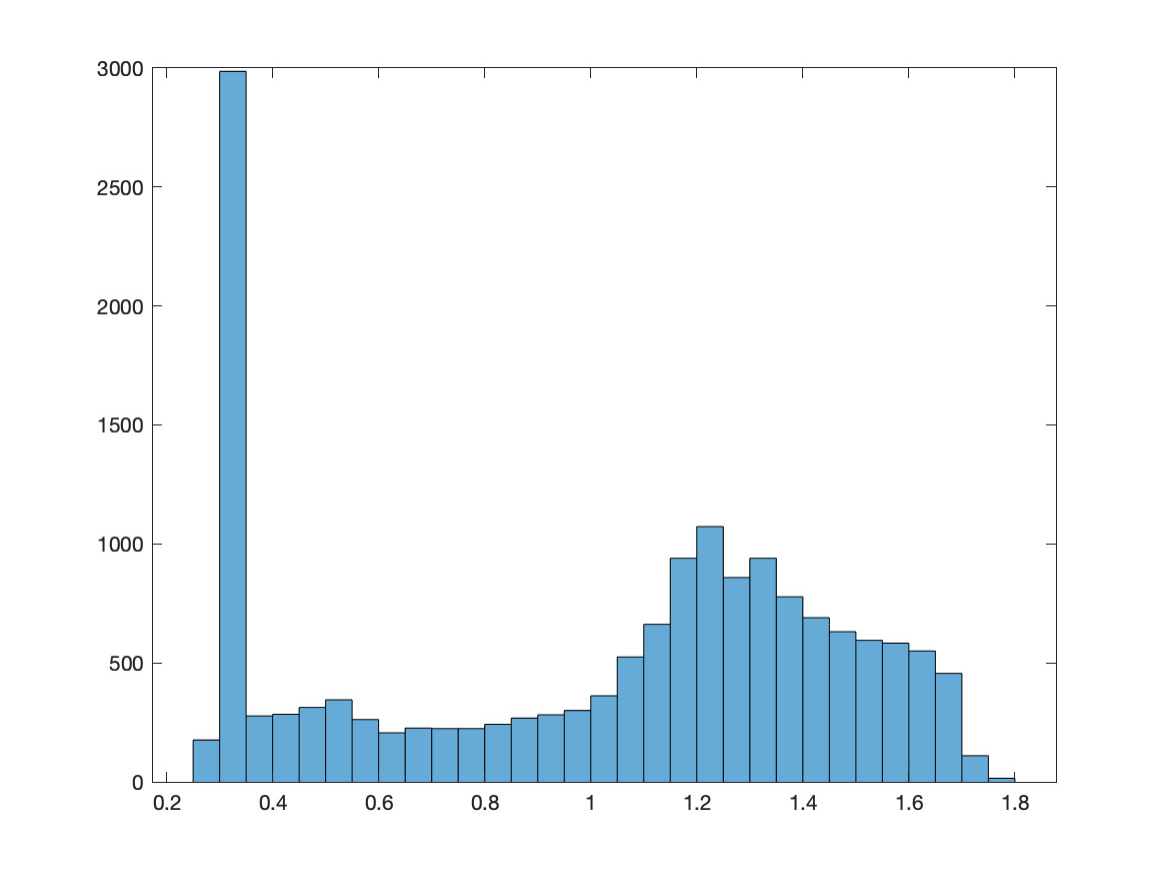}}
\subfigure[]{
\includegraphics[width=1.4in,height=1.4in]{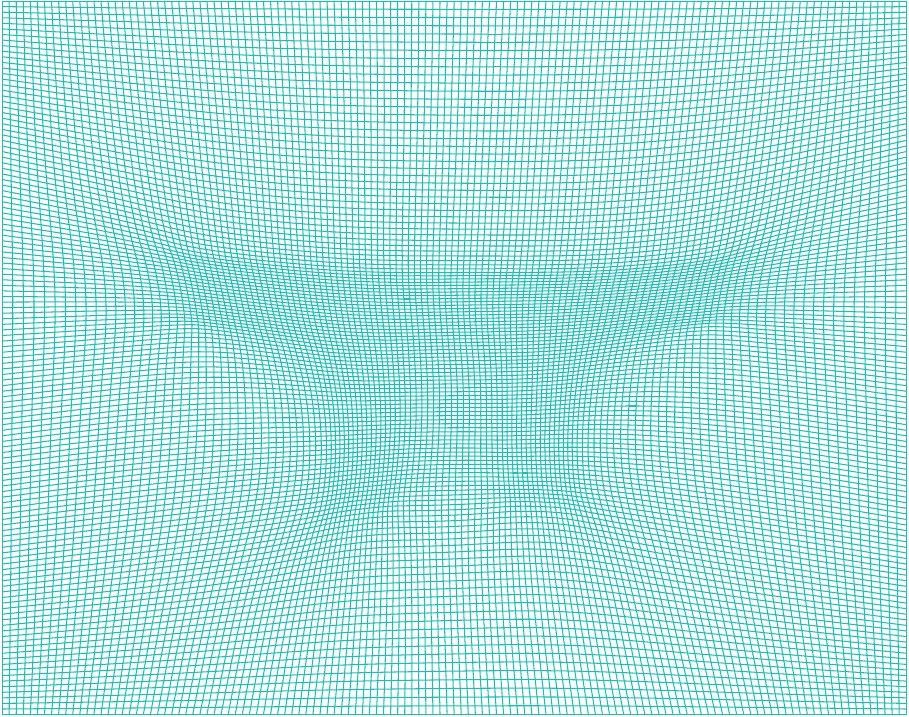}
\includegraphics[width=1.4in,height=1.4in]{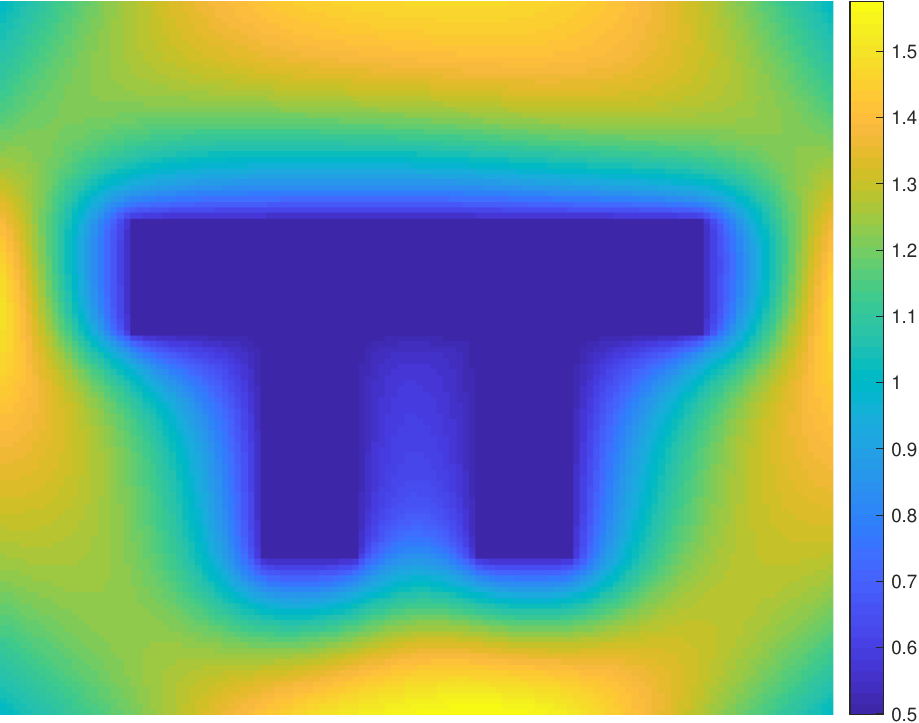}
\includegraphics[width=1.4in,height=1.4in]{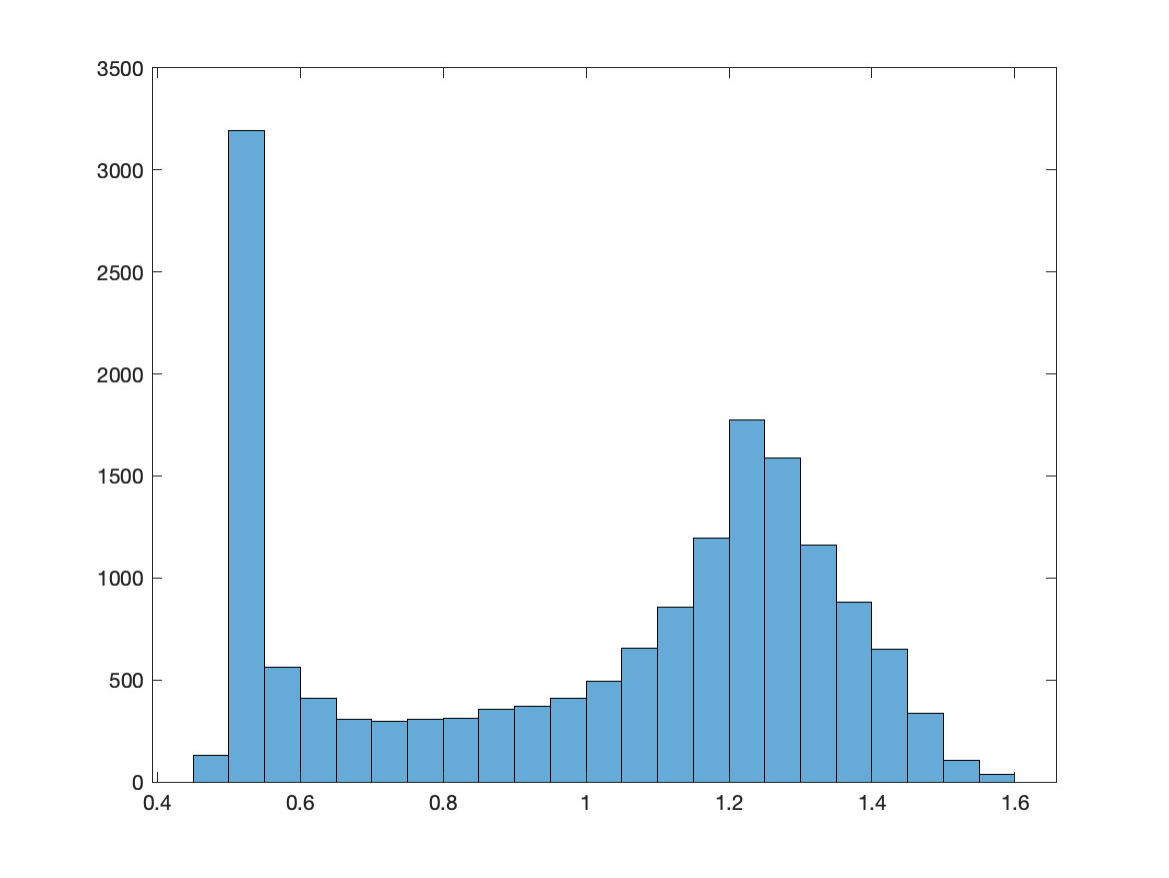}}
\subfigure[]{
\includegraphics[width=1.4in,height=1.4in]{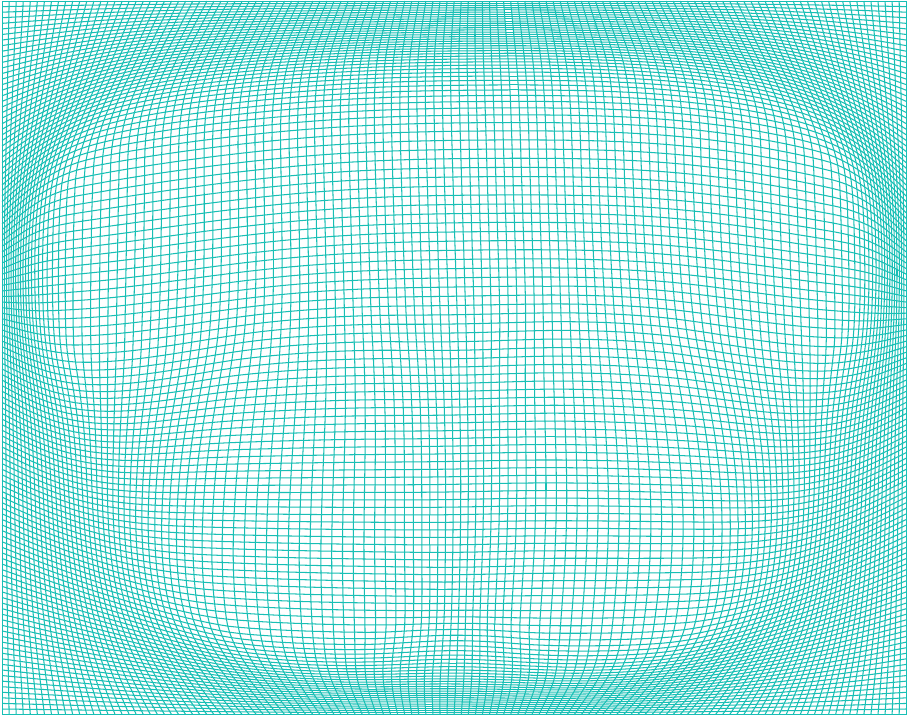}
\includegraphics[width=1.4in,height=1.4in]{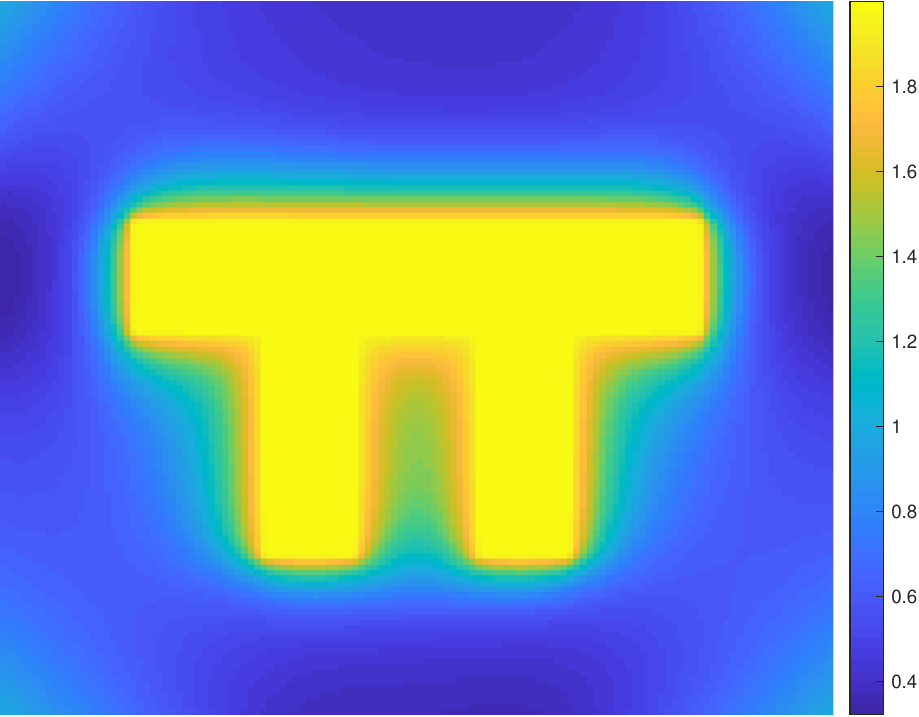}
\includegraphics[width=1.4in,height=1.4in]{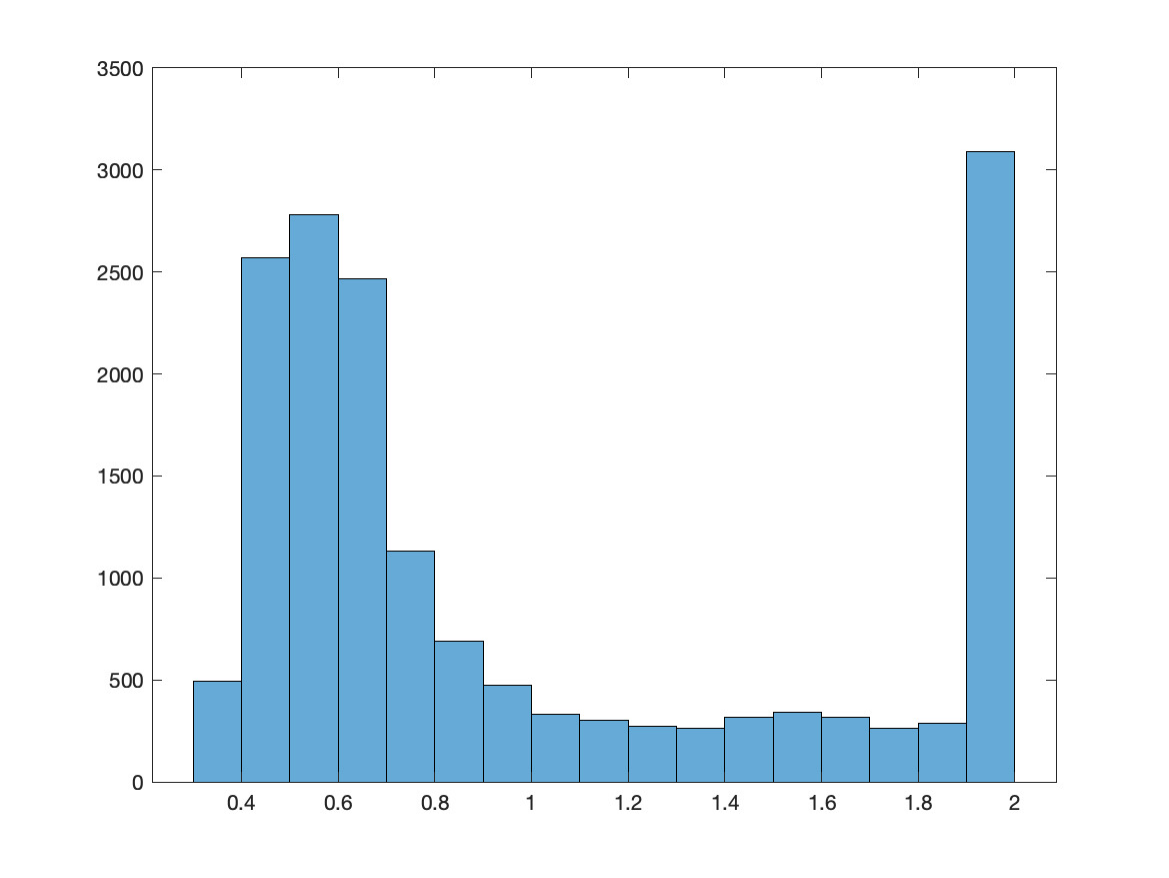}}
\subfigure[]{
\includegraphics[width=1.4in,height=1.4in]{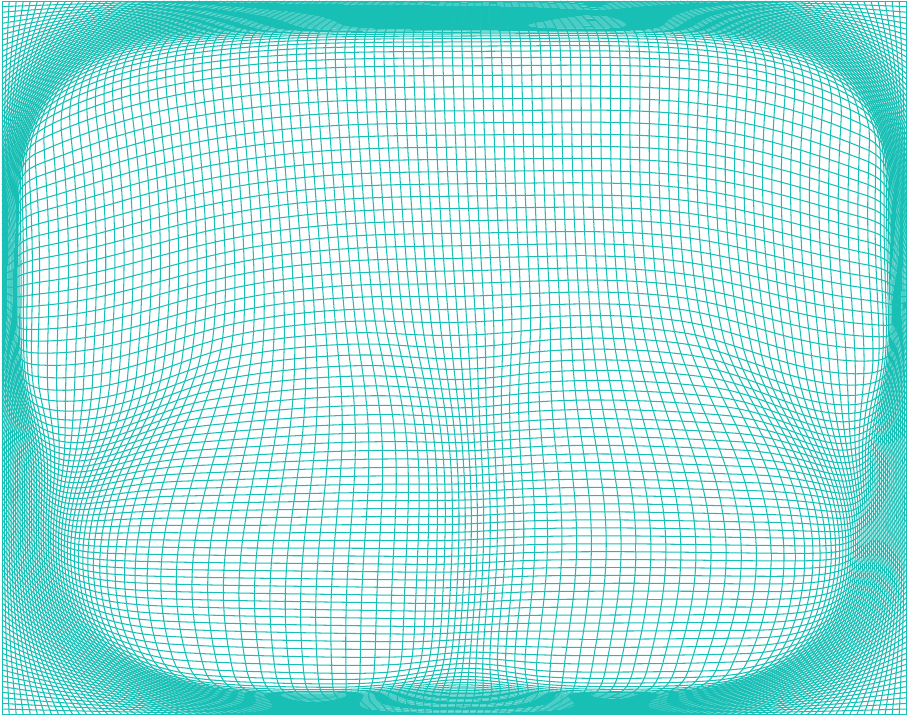}
\includegraphics[width=1.4in,height=1.4in]{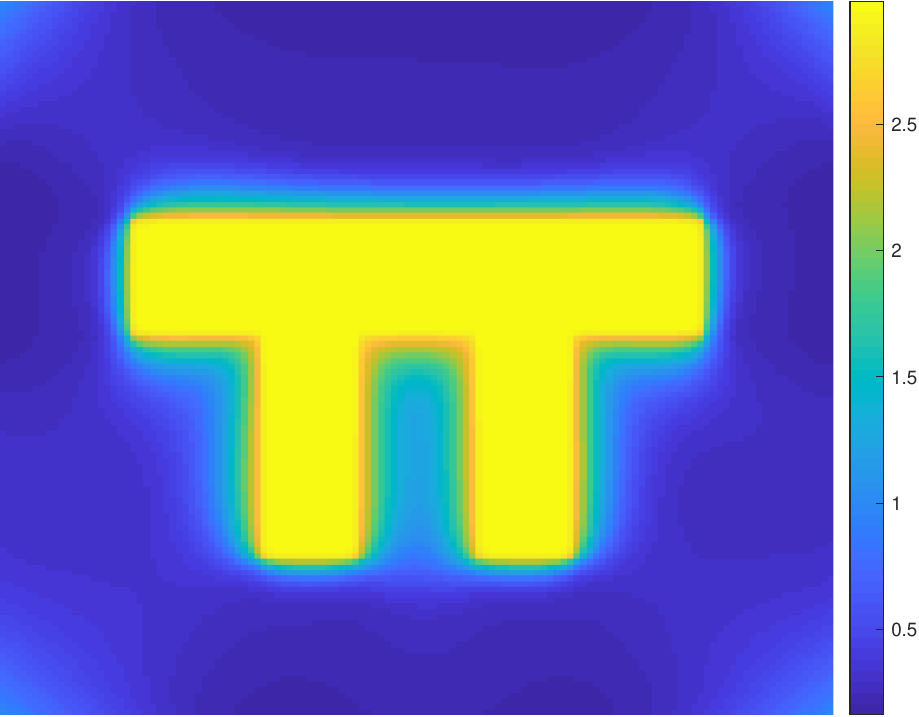}
\includegraphics[width=1.4in,height=1.4in]{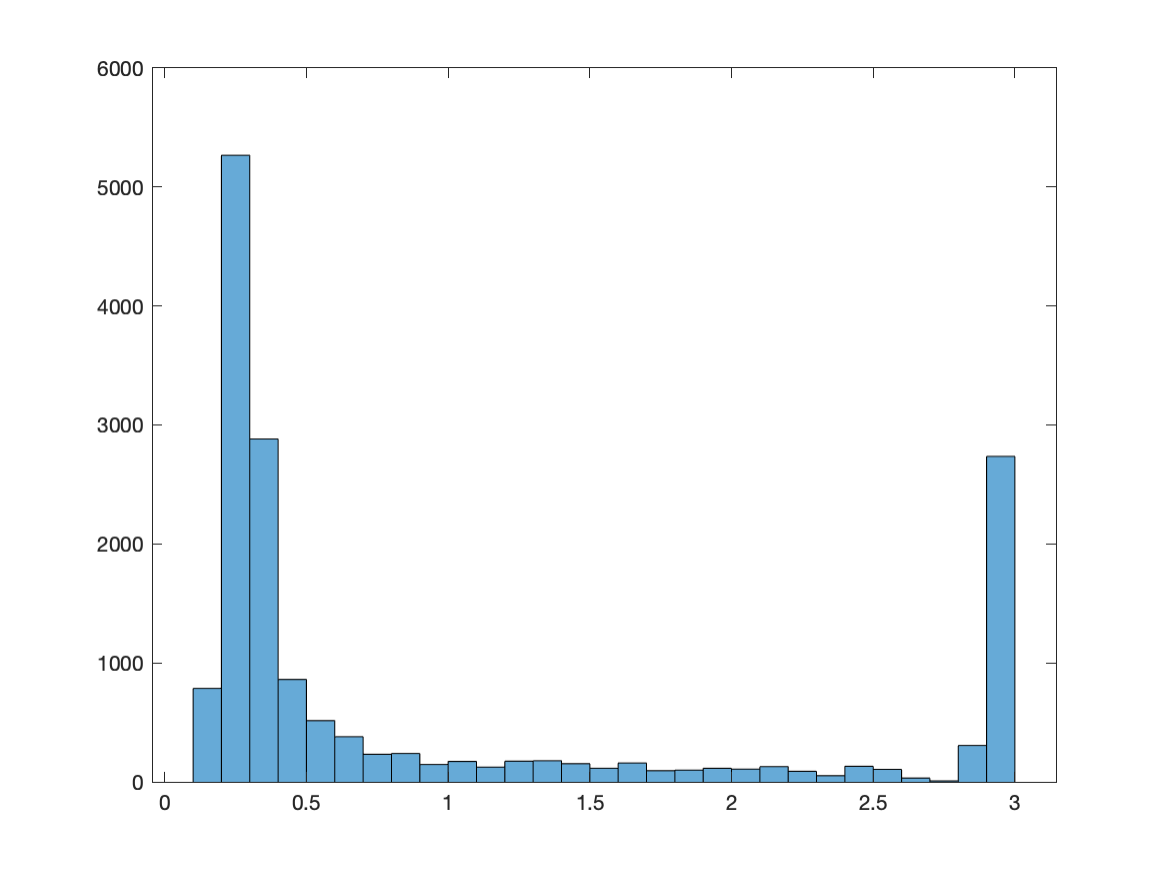}}
\caption{An example of quasi-conformal mapping with volume prior and optimized volumetric distortion in 2D. (a) A 2D domain with a prescribed $\pi$-shaped region. (b)--(e) The results with volume prior $\theta=\ln{0.3}$, $\ln{0.5}$, $\ln{2}$ and $\ln{3}$ for the specific region $\Omega'$. For each volume prior, the resulting transformation, the Jacobian determinant map and its corresponding histograms are shown.}\label{fig:case3_2D}
\end{figure}

\begin{figure}[t]
    \centering
\includegraphics[width=\textwidth]{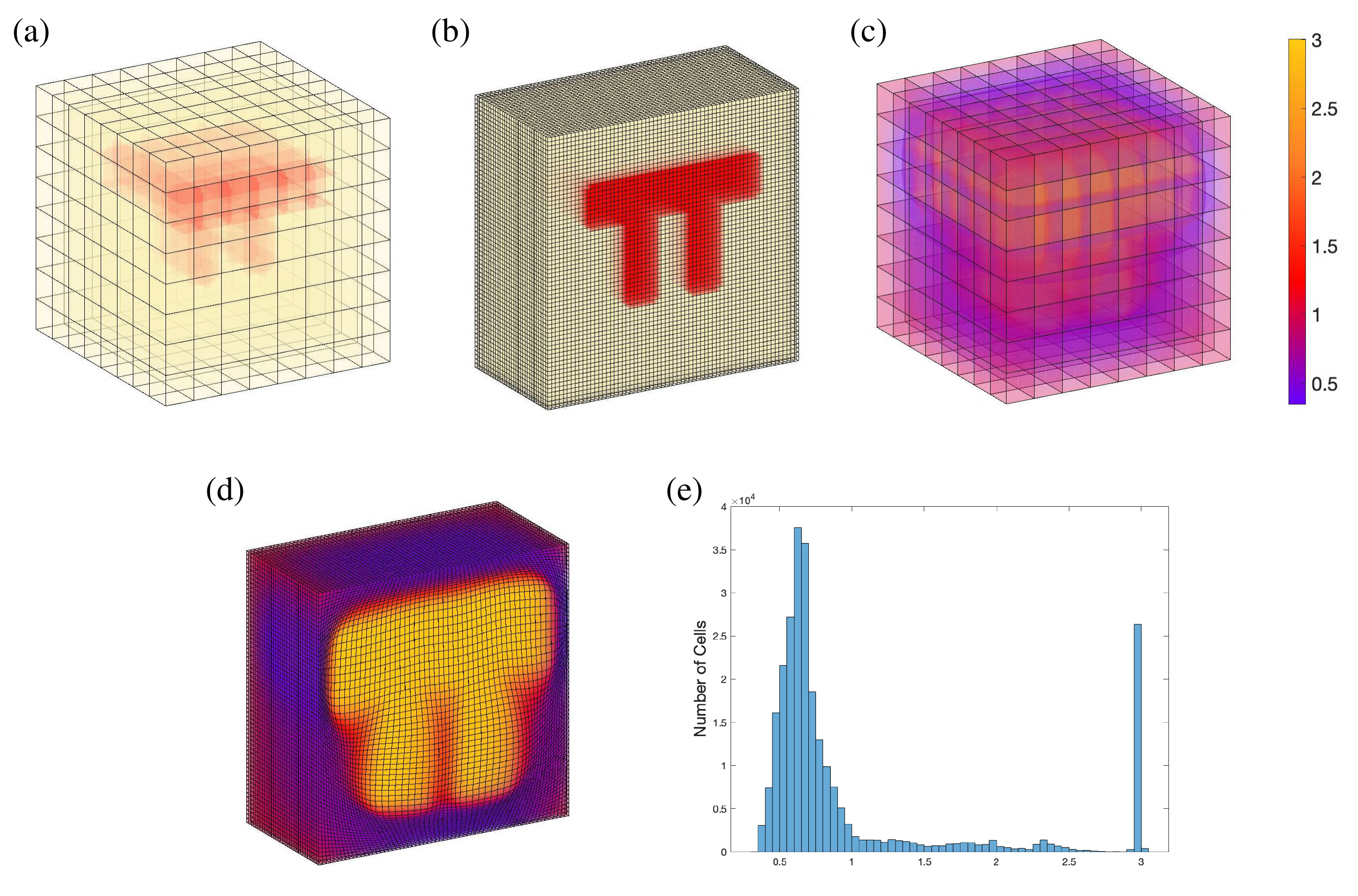}
    \caption{An adaptive remeshing example in 3D. (a) The initial volumetric domain, with the region for which we specify a volume prior of $\ln(3)$ highlighted in red. For better visualization, only a part of the grid lines are shown. (b) A cross-sectional view of the domain in (a) with all grid lines shown. (c) The 3D quasi-conformal mapping result color-coded with the Jacobian determinant. For better visualization, only a part of the grid lines are shown. (d) A cross-sectional view of the result in (c) with all grid lines shown. (e) A histogram of the Jacobian determinant.}
    \label{fig:case3_3D}
\end{figure}

\subsection{$n$-dimensional quasi-conformal mapping with volume prior and optimized volumetric distortion}
We now test the proposed model \eqref{case3}. In the following experiments, we set $\alpha_{1} = 1$, $\alpha_{2}=1$, $\alpha_{3}=0.1$ and $\alpha_{4} = 10^{5}$ in \eqref{case3}.

As shown in Fig.~\ref{fig:case3_2D}(a), here we consider a $\pi$-shaped region (denoted as $\Omega'$) in a 2D domain and set different volume priors for the specific region to produce different deformations. Here, we use four different volume priors: $\ln(0.3)$, $\ln(0.5)$, $\ln(2)$ and $\ln(3)$, with results presented in Fig.~\ref{fig:case3_2D}(b)--(e). For each experiment, the corresponding transformation, the Jacobian determinant map and the histogram of the Jacobian determinant are also shown. It can be observed that in all cases, the transformations are folding-free. The Jacobian determinant maps also illustrate that the proposed model \eqref{PM} indeed produce the desired transformations that satisfy the volume priors.

We then test the proposed model using another synthetic example in 3D. For this example, we again set a volume prior of $\ln(3)$ for a specific region in a solid domain (see Fig.~\ref{fig:case3_3D}(a)--(b)) and solve the proposed model to obtain a 3D transformation. This results in a deformed volumetric domain with the specified region is significantly expanded (see Fig.~\ref{fig:case3_3D}(c)--(d)). As shown in the histogram of the Jacobian determinant in Fig.~\ref{fig:case3_3D}(e), the quasi-conformal mapping satisfies the volume prior accurately.

\subsection{The most general model}
Here, we use a synthetic 3D example to highlight the advantage of the general model \eqref{PM}. As shown in Fig.~\ref{fig:general_model}(a)--(b), the template and reference are both cubes but the reference misses some parts, which represent occlusions. The eight outermost corner vertices are chosen as the landmarks. Fig.~\ref{fig:general_model}(c)--(f) show the mapping results using different models, from which we can clearly see that the general model \eqref{PM} gives the best result (Fig.~\ref{fig:general_model}(c)). More specifically, note that the intensity term helps match the deformed template and the reference. However, because of the missing part in the reference, the intensity term may cause some overfitting in matching the two shapes. By introducing the volume prior, we can prevent such overfitting effectively. Finally, the landmarks are useful for registering the salient features of the two shapes. As shown in Fig.~\ref{fig:general_model}(d)--(f), missing any of these components leads to an unsatisfactory result.

\begin{figure}[t!]
\centering
\includegraphics[width=\textwidth]{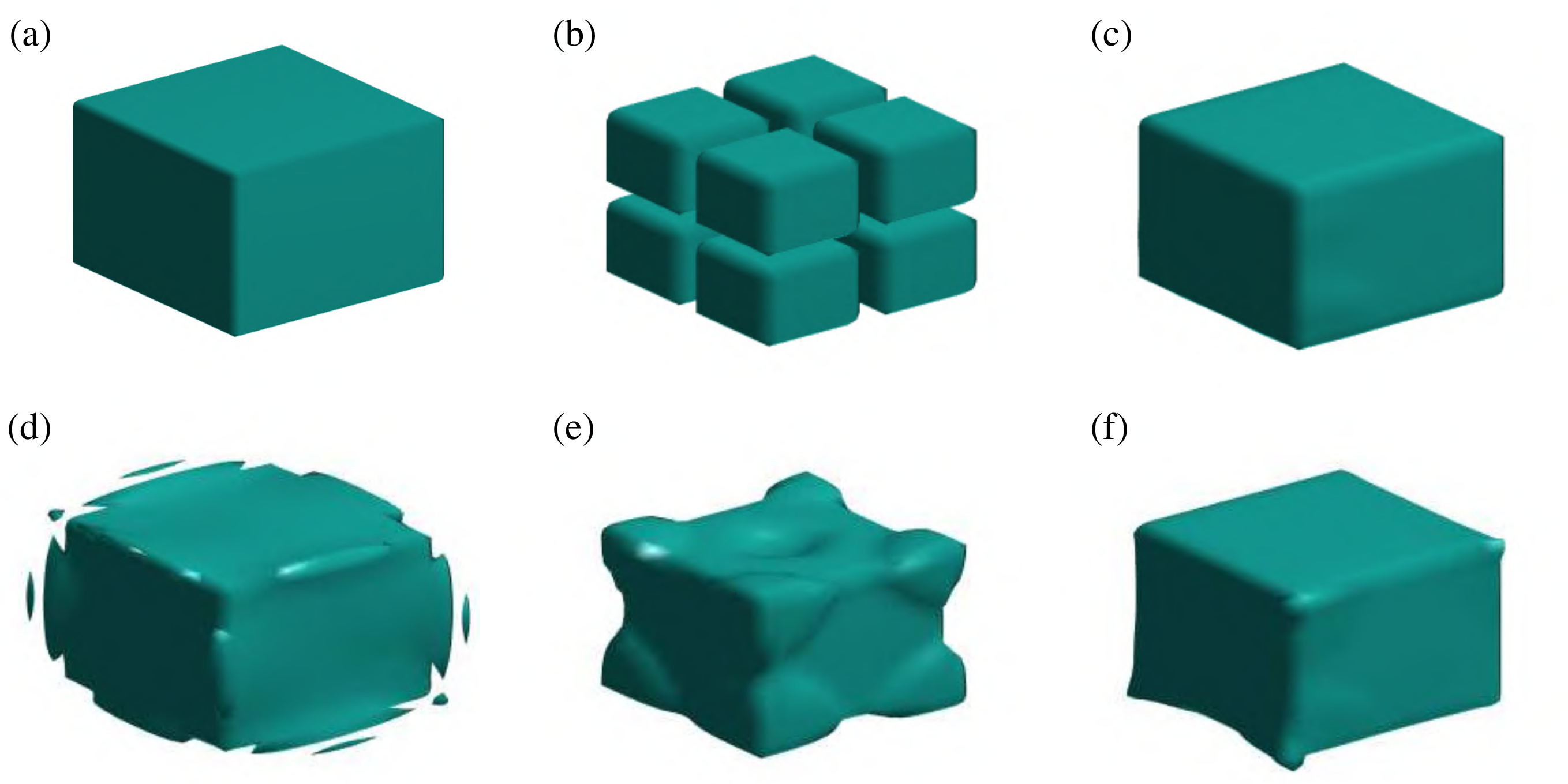}
\caption{An 3D example of quasi-conformal mapping with intensity, volume prior and landmarks. (a) The template. (b) The reference. (c) The registration result obtained using the general model\eqref{PM}. (d)--(f) The registration result obtained using the models without intensity, volume prior, and landmarks respectively. Here, the landmarks are the eight corner vertices of the cubes.}\label{fig:general_model}
\end{figure}

\begin{figure}[t!]
\centering
\includegraphics[width=0.95\textwidth]{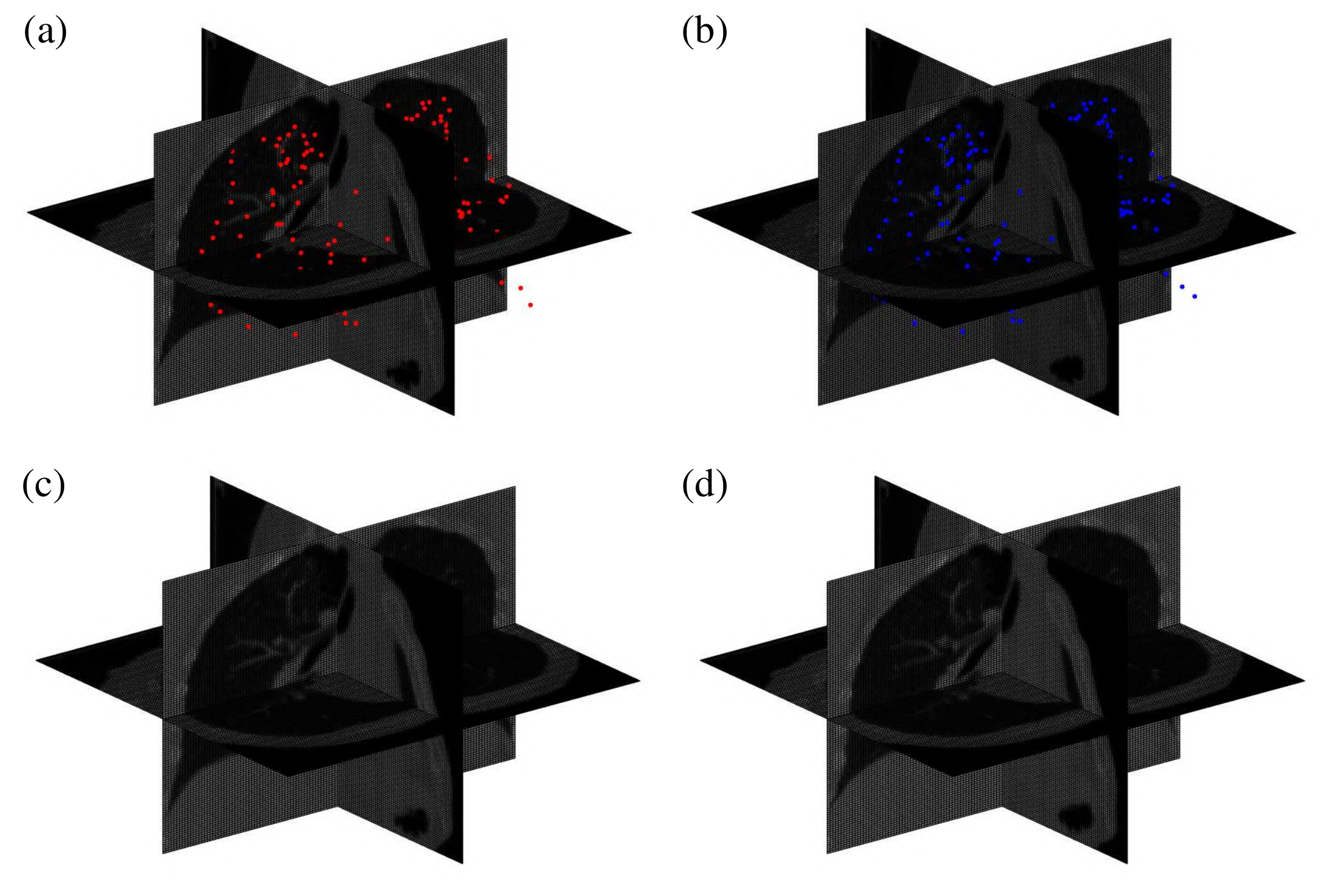}
\caption{Registration of 3D lung images. (a) The template image with landmarks. (b) The reference image with landmarks. (c) The deformed template generated by our proposed model with landmarks only (i.e. \eqref{case1}). (d) The deformed template generated by our model with both landmark and intensity considered (i.e. \eqref{case2}).}\label{fig:Lung1}
\end{figure}

\section{Applications} \label{sect:application}
After evaluating the performance of the proposed unifying framework using the above examples, we explore the applications of the framework in medical imaging, engineering and graphics.

\subsection{Medical image registration}
The prior landmark-constrained $n$-dimensional quasi-conformal mapping method~\cite{lee2016landmark} has been shown to be useful for medical image registration. Analogously, our proposed framework can be effectively applied to the registration of $n$-dimensional medical images. More specifically, given two medical images and some common prominent features in each of them, we can compute a landmark-matching quasi-conformal mapping between them using the proposed model in~\eqref{case1}. In addition, we also incorporate the intensity information and find a landmark- and intensity-matching quasi-conformal mapping by using the proposed model in~\eqref{case2}.

We test \eqref{case1} and \eqref{case2} using lung CT data with 300 pairs of landmarks as shown in Fig.~\ref{fig:Lung1}(a)--(b), which is freely available from the Deformable Image Registration Laboratory (\url{www.dir-lab.com}). We rescale the data into the size of $128\times128\times128$. For the parameters, we set $\alpha_{2} =1, \alpha_{3}=0.01$ for \eqref{case1} and $\alpha_{2} =1, \alpha_{3}=0.01, \alpha_{5} = 0.1$ for \eqref{case2}. The registration results for the two cases are displayed in Fig.~\ref{fig:Lung1}(c)--(d). Both of the registration results satisfy the landmark constraints, with the landmark mismatch error less than $10^{-6}$. This indicates that both models are capable of matching prescribed landmark features. To better compare the two registration results, we consider the axial, sagittal, and coronal views as shown in Fig.~\ref{fig:Lung2}. The reference and the template from the three views are first provided in Fig.~\ref{fig:Lung2}(a)--(b). We then consider their intensity difference as shown in Fig.~\ref{fig:Lung2}(c) and quantify the intensity mismatch by defining
\begin{equation}
    Re_{SSD} = \frac{\|T(\bm{y})-R\|_{2}^{2}}{\|T-R\|_{2}^{2}} \times 100\%,
\end{equation}
where $T$ and $R$ denote the template and reference images, and $\bm{y}$ is the deformation. This gives $Re_{SSD} = 100\%$ for the original difference (with $\bm{y} = Id$). For the registration result obtained by the landmark-based model~\eqref{case1} (Fig.~\ref{fig:Lung2}(d)--(e)), it is clear that the intensity is not well matched, with $Re_{SSD} = 75.74\%$. By contrast, for the result obtained by the landmark- and intensity-based model~\eqref{case2} (Fig.~\ref{fig:Lung2}(f)--(g)), it can be observed that the intensity mismatch is low (with $Re_{SSD} = 9.10\%$). This demonstrates the effectiveness of our framework for medical image registration.

\begin{figure}[t!]
\subfigure[Reference from three views]{
\includegraphics[width=0.95in,height=0.95in]{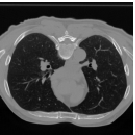}
\includegraphics[width=0.95in,height=0.95in]{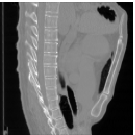}
\includegraphics[width=0.95in,height=0.95in]{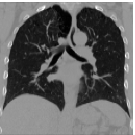}} \\
\subfigure[Template from  three views]{
\includegraphics[width=0.95in,height=0.95in]{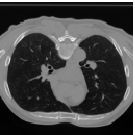}
\includegraphics[width=0.95in,height=0.95in]{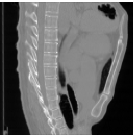}
\includegraphics[width=0.95in,height=0.95in]{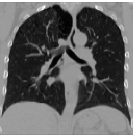}} 
\subfigure[Difference between $T$ and $R$, $Re_{SSD} = 100\%$]{
\includegraphics[width=0.95in,height=0.95in]{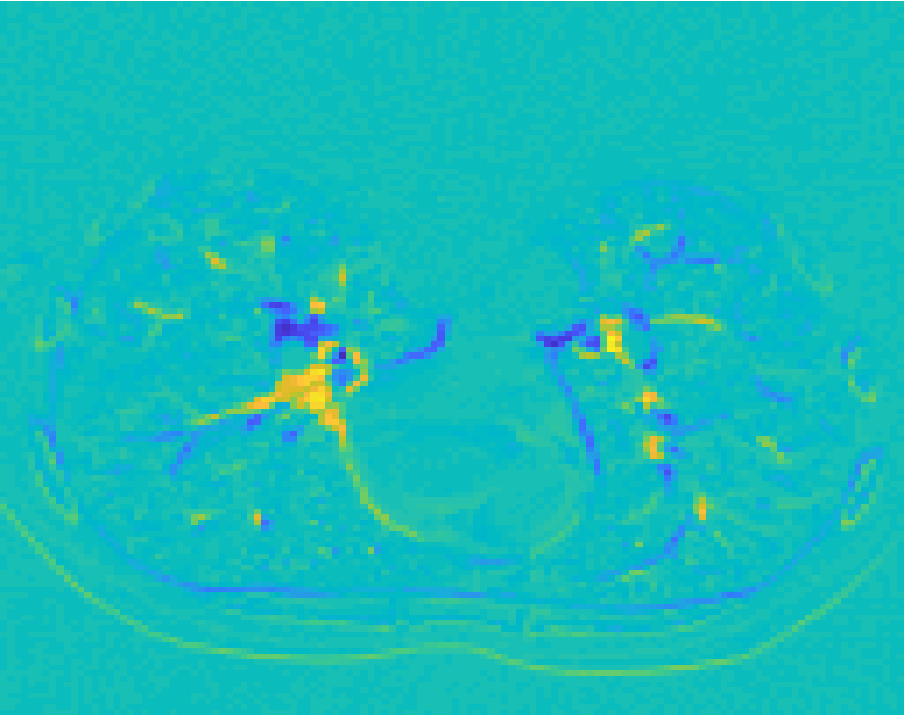}
\includegraphics[width=0.95in,height=0.95in]{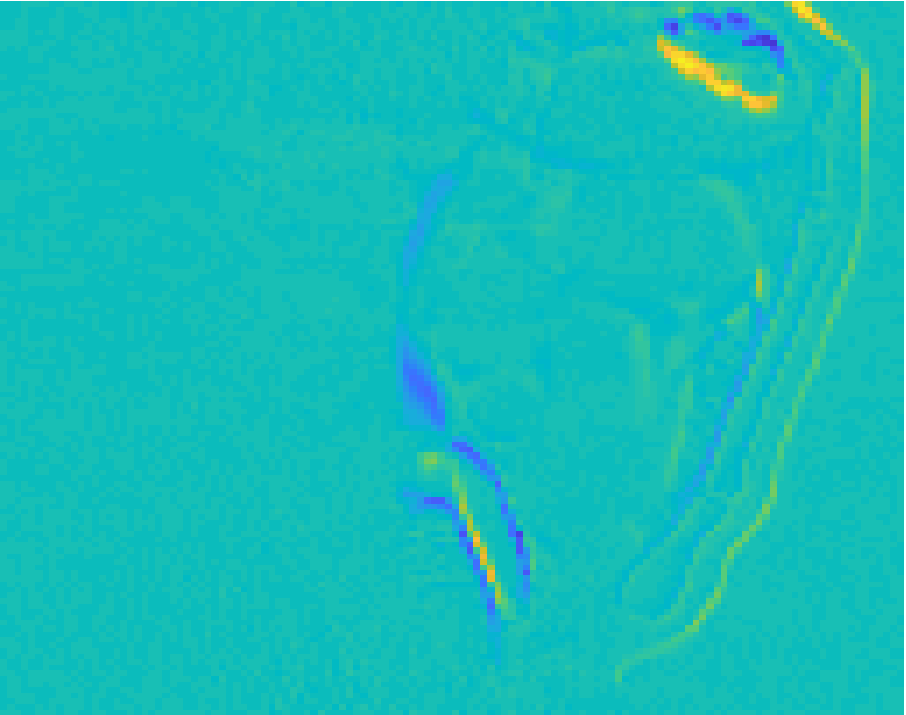}
\includegraphics[width=0.95in,height=0.95in]{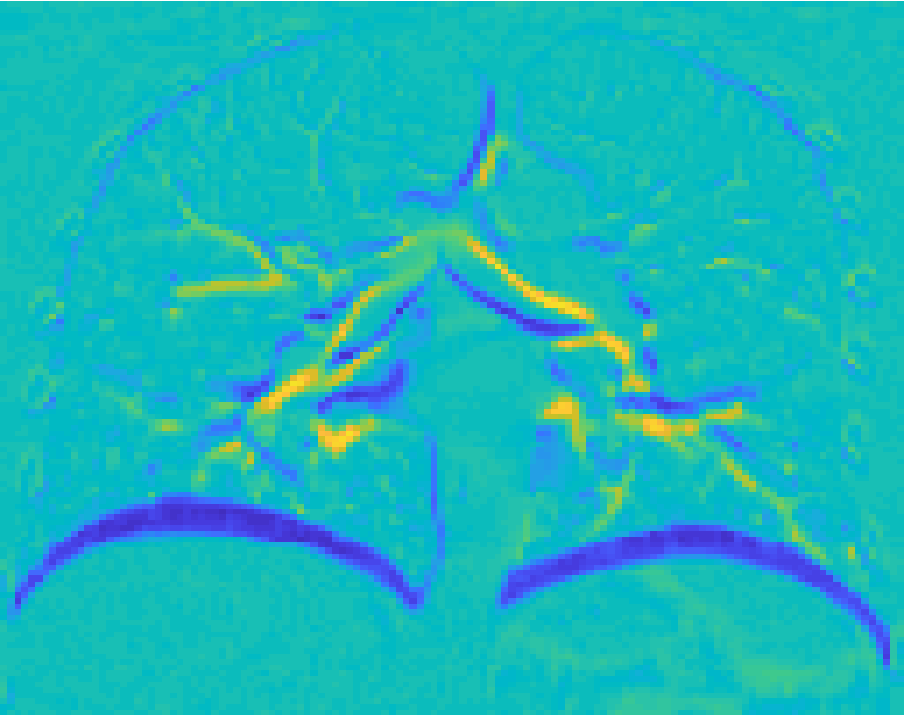}}  \\
\subfigure[$T(\bm{y})$ with landmarks from  three views]{
\includegraphics[width=0.95in,height=0.95in]{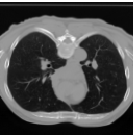}
\includegraphics[width=0.95in,height=0.95in]{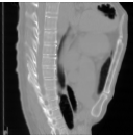}
\includegraphics[width=0.95in,height=0.95in]{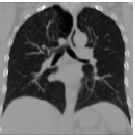}} 
\subfigure[Difference between $T(\bm{y})$ and $R$, $Re_{SSD} = 75.74\%$]{
\includegraphics[width=0.95in,height=0.95in]{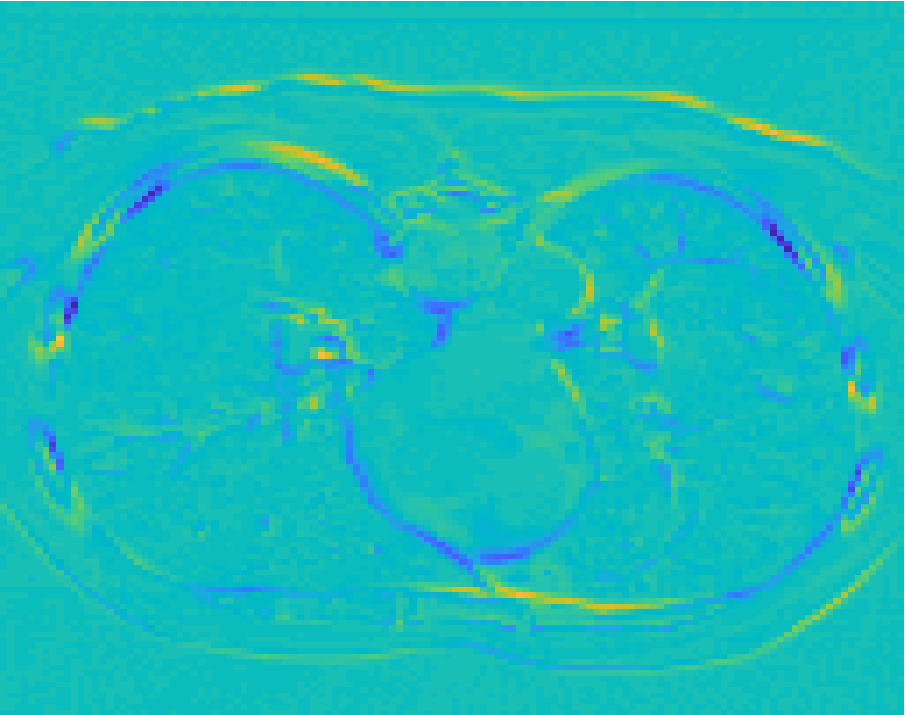}
\includegraphics[width=0.95in,height=0.95in]{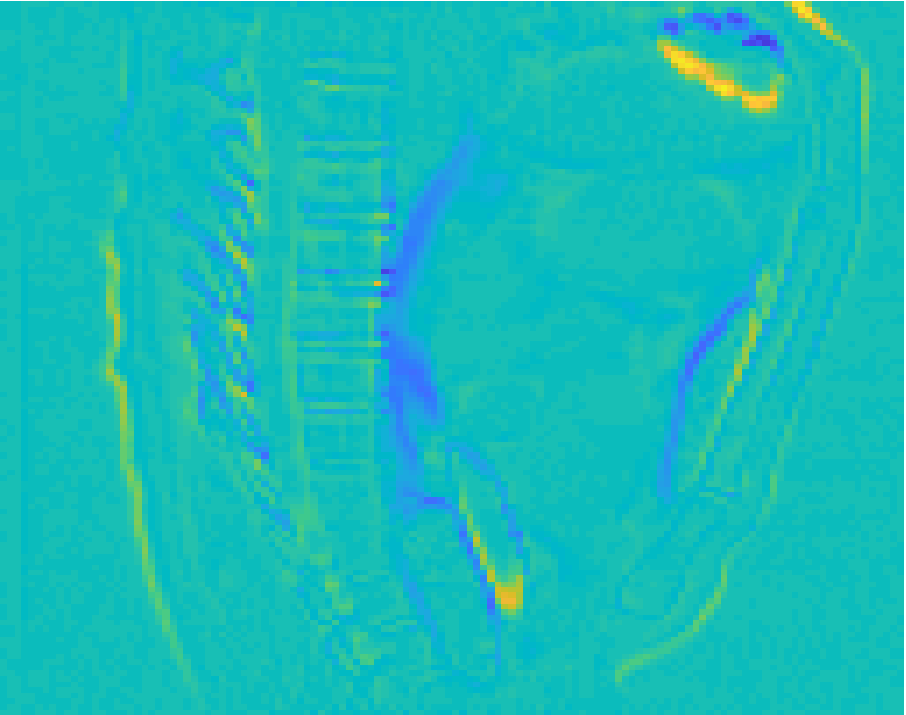}
\includegraphics[width=0.95in,height=0.95in]{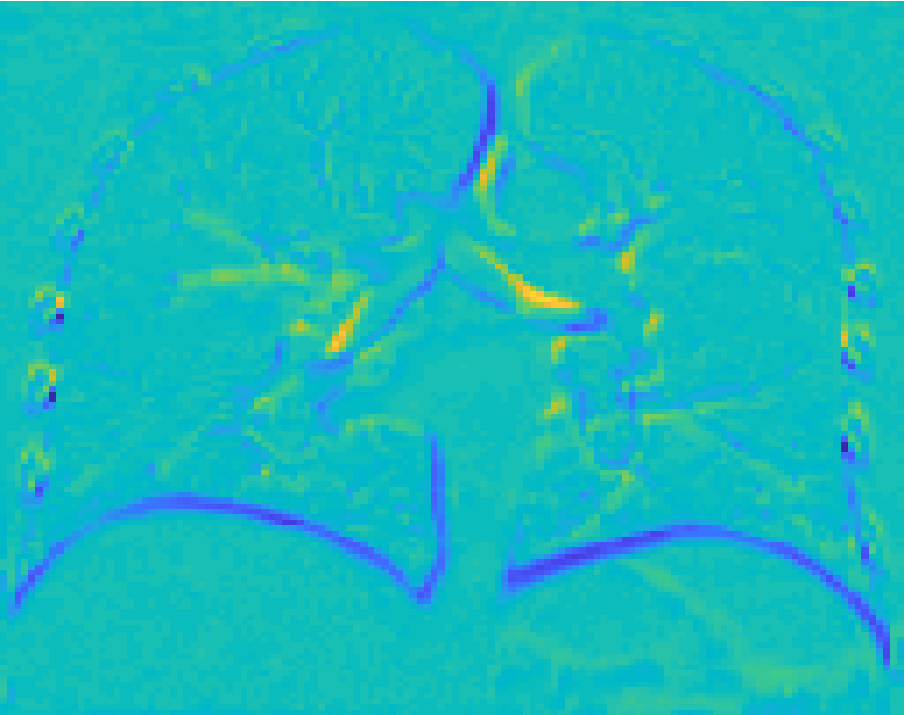}} \\
\subfigure[$T(\bm{y})$ with landmarks and intensity from three views]{
\includegraphics[width=0.95in,height=0.95in]{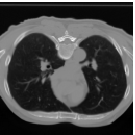}
\includegraphics[width=0.95in,height=0.95in]{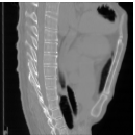}
\includegraphics[width=0.95in,height=0.95in]{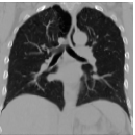}} 
\subfigure[Difference between $T(\bm{y})$ and $R$, $Re_{SSD} = 9.10\%$]{
\includegraphics[width=0.95in,height=0.95in]{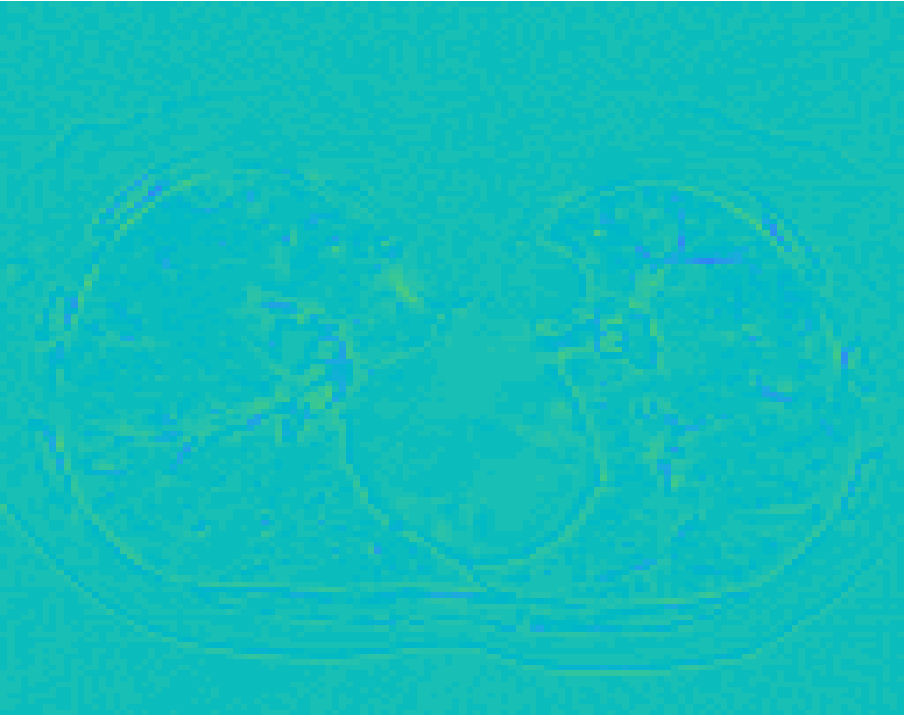}
\includegraphics[width=0.95in,height=0.95in]{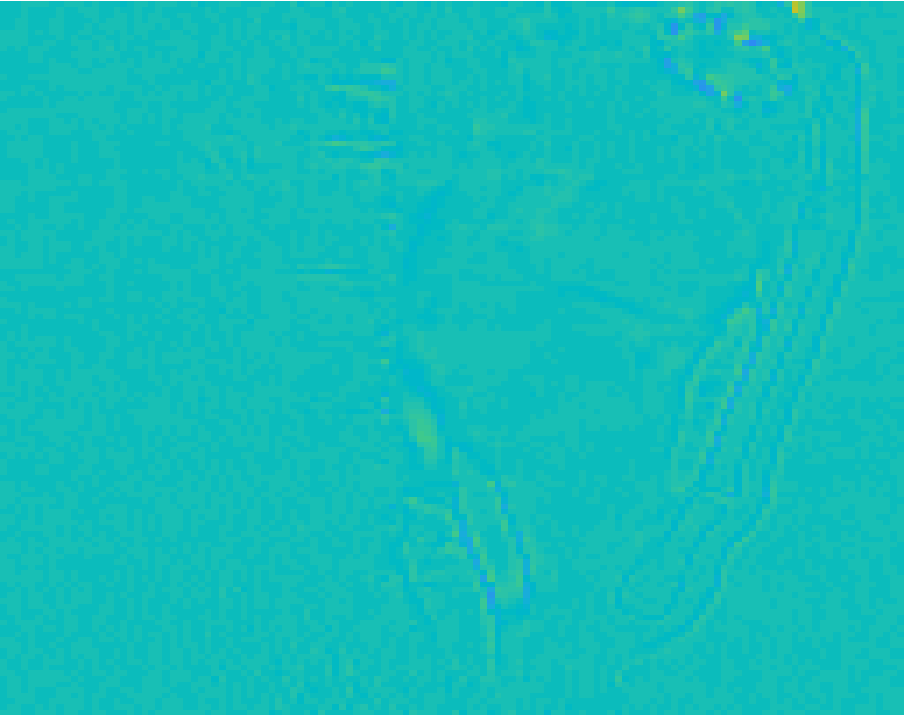}
\includegraphics[width=0.95in,height=0.95in]{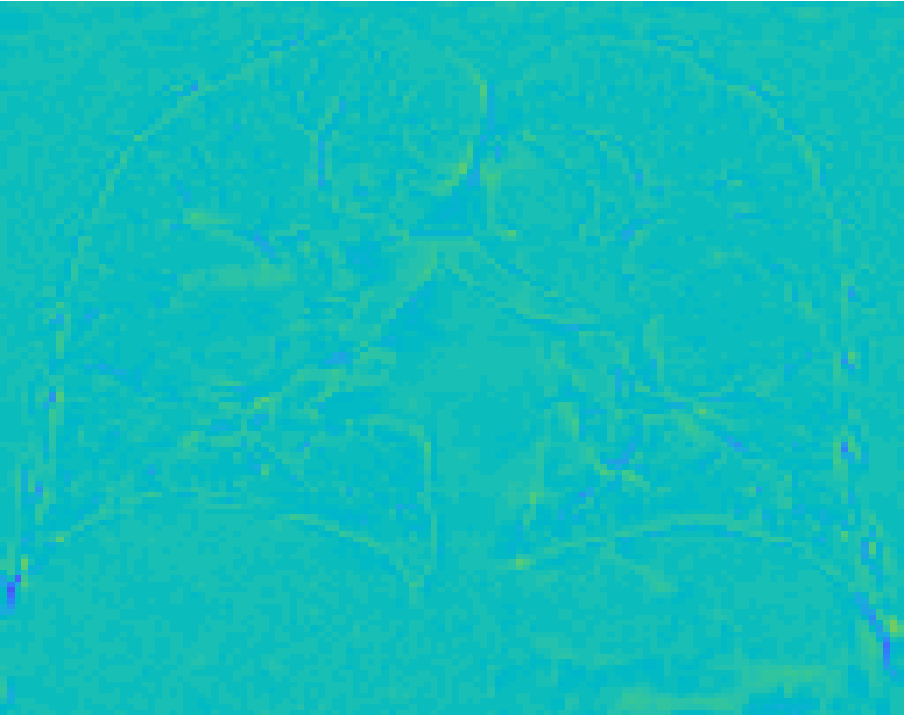}} 
\caption{Visualization of the 3D lung image registration results via the axial view, sagittal view, and coronal view. The difference between the reference and the template, and the difference between the reference and the two registration results are both provided. Note that both registration models \eqref{case1} and \eqref{case2} are capable of matching the features, and the more advanced model \eqref{case2} with the intensity fidelity term added gives a more accurate registration.}\label{fig:Lung2}
\end{figure}

\subsection{Adaptive remeshing}
For many applications that involve solving differential equations, it is crucial to remesh an object in an adaptive manner so that certain parts of it are with a finer mesh resolution. Using our proposed framework, adaptive remeshing in any $n$-dimensional domain can be easily achieved. More specifically, to remesh an object such that different regions of it have different mesh density, we can set an area/volume prior for each specific region and utilize the proposed model in~\eqref{case3} to find a quasi-conformal transformation. 

We illustrate this idea using the examples presented in Fig.~\ref{fig:case3_2D} and Fig.~\ref{fig:case3_3D}. After finding the resulting transformations, the remeshing result can be constructed by inverting them. Figure \ref{fig:synthetic_remeshing}(a)--(d) shows the adaptive remeshing results corresponding to the four volume priors in Fig.~\ref{fig:case3_2D}(b)--(e). We can see that different volume priors can lead to either a finer mesh in the prescribed region when $\theta>0$ or a coarser mesh in the prescribed region when $\theta<0$. Analogously, we achieve an adaptive remeshing result of a 3D solid domain with a higher mesh density at the specified region by generating a regular grid and transforming it using the inverse of the 3D quasi-conformal mapping (see Fig.~\ref{fig:synthetic_remeshing}(e)). 

\begin{figure}[t]
    \centering
\subfigure[]{
\includegraphics[width=1.8in,height=1.8in]{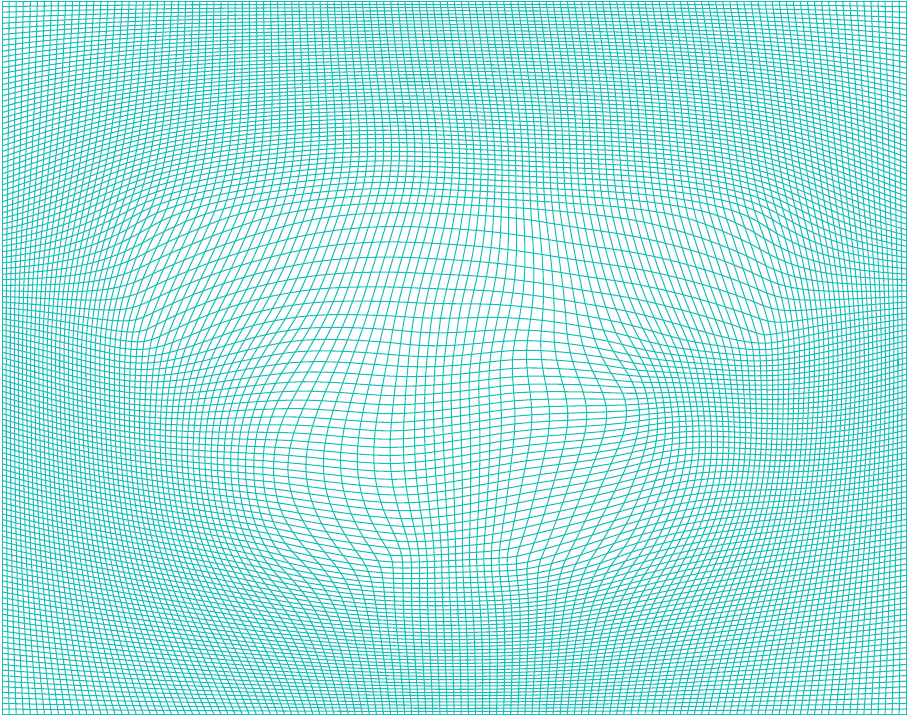}}
\subfigure[]{
\includegraphics[width=1.8in,height=1.8in]{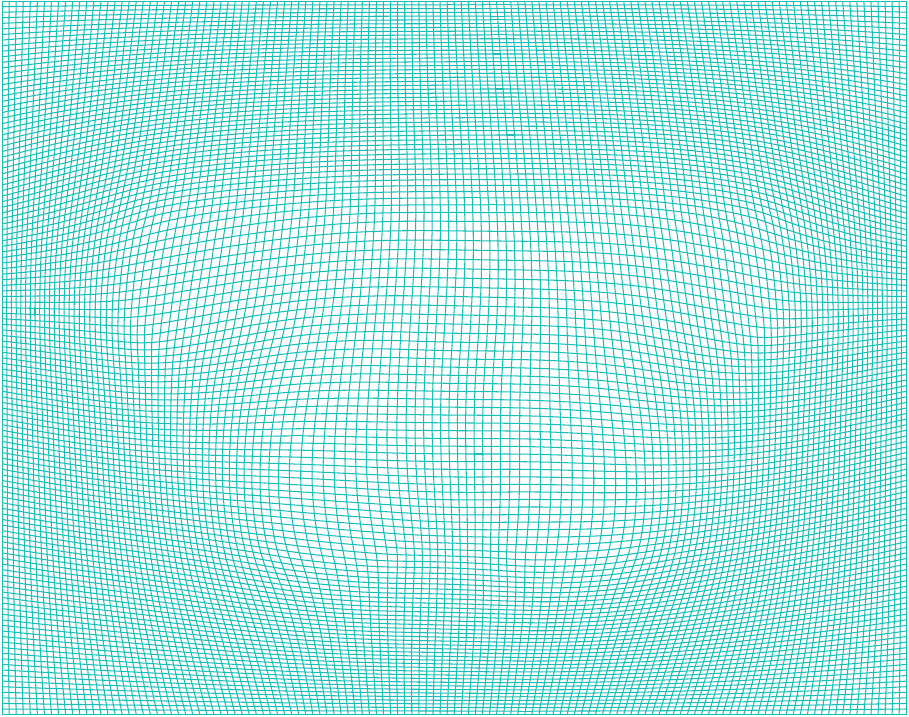}}
\subfigure[]{
\includegraphics[width=1.8in,height=1.8in]{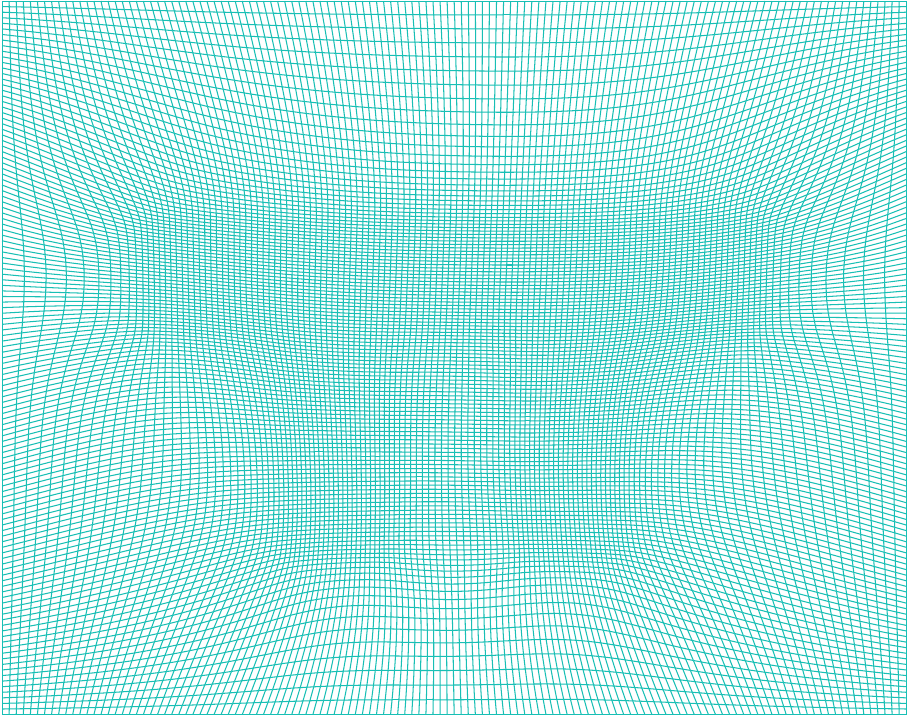}}
\subfigure[]{
\includegraphics[width=1.8in,height=1.8in]{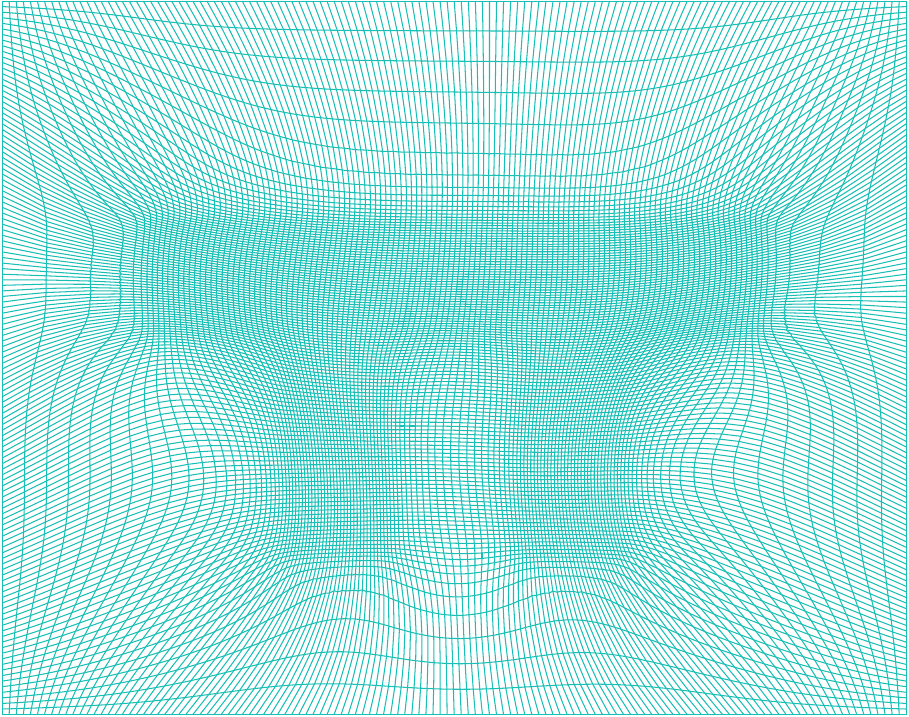}}
\subfigure[]{
\includegraphics[width=0.45\textwidth]{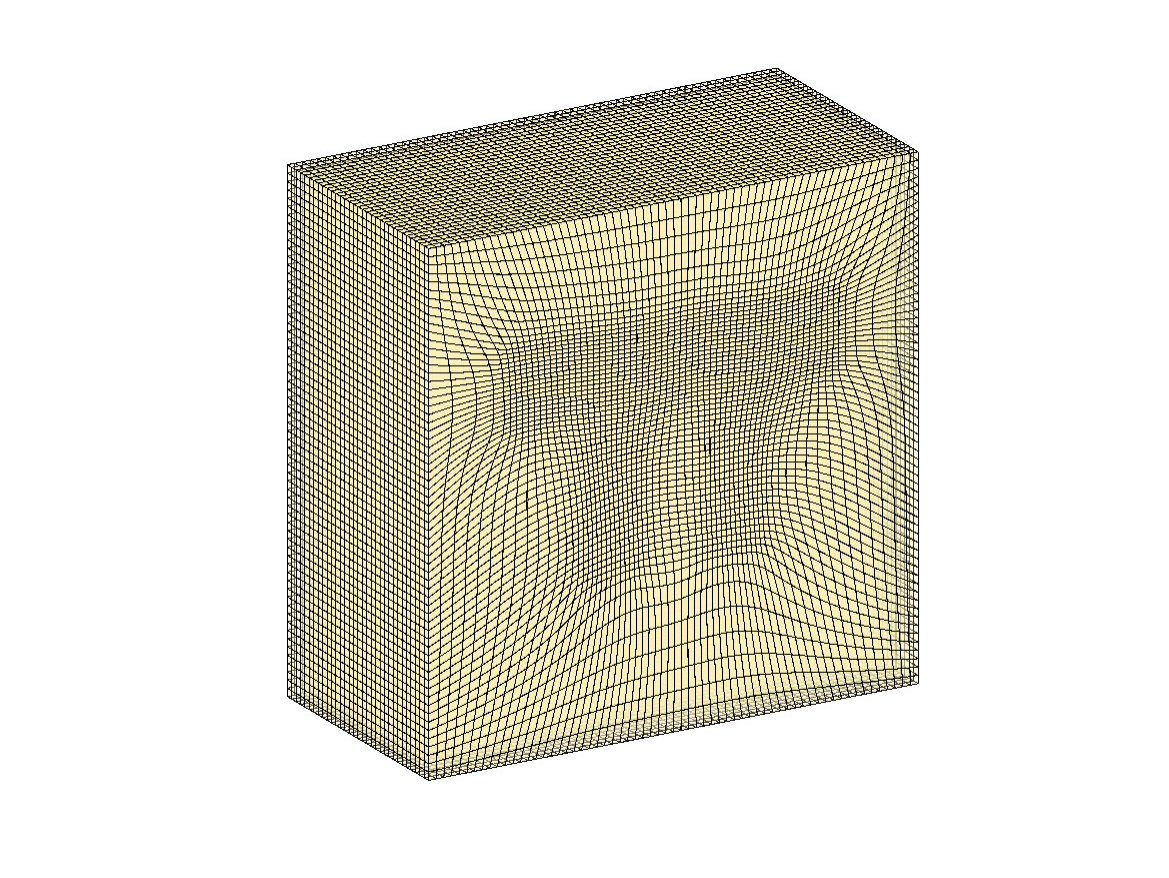}}
    \caption{Adaptive remeshing using our proposed model~\eqref{case3}. (a)--(d) The adaptive remeshing results achieved using the transformations in Fig.~\ref{fig:case3_2D}(b)--(e) respectively. (e) The adaptive remeshing result achieved using the transformation in Fig.~\ref{fig:case3_3D}.}
    \label{fig:synthetic_remeshing}
\end{figure}

\subsection{Deformation-based shape modeling}
Mesh deformations are frequently used in computer graphics for achieving different shape modelling and animation effects~\cite{zhou2005large,hildebrandt2011interactive}. As our proposed model is capable of producing quasi-conformal mappings in $n$-dimensional space, which are smooth and folding-free, it is well-suited for deformation-based shape modeling.

\begin{figure}[t!]
\centering
\includegraphics[width=0.95\textwidth]{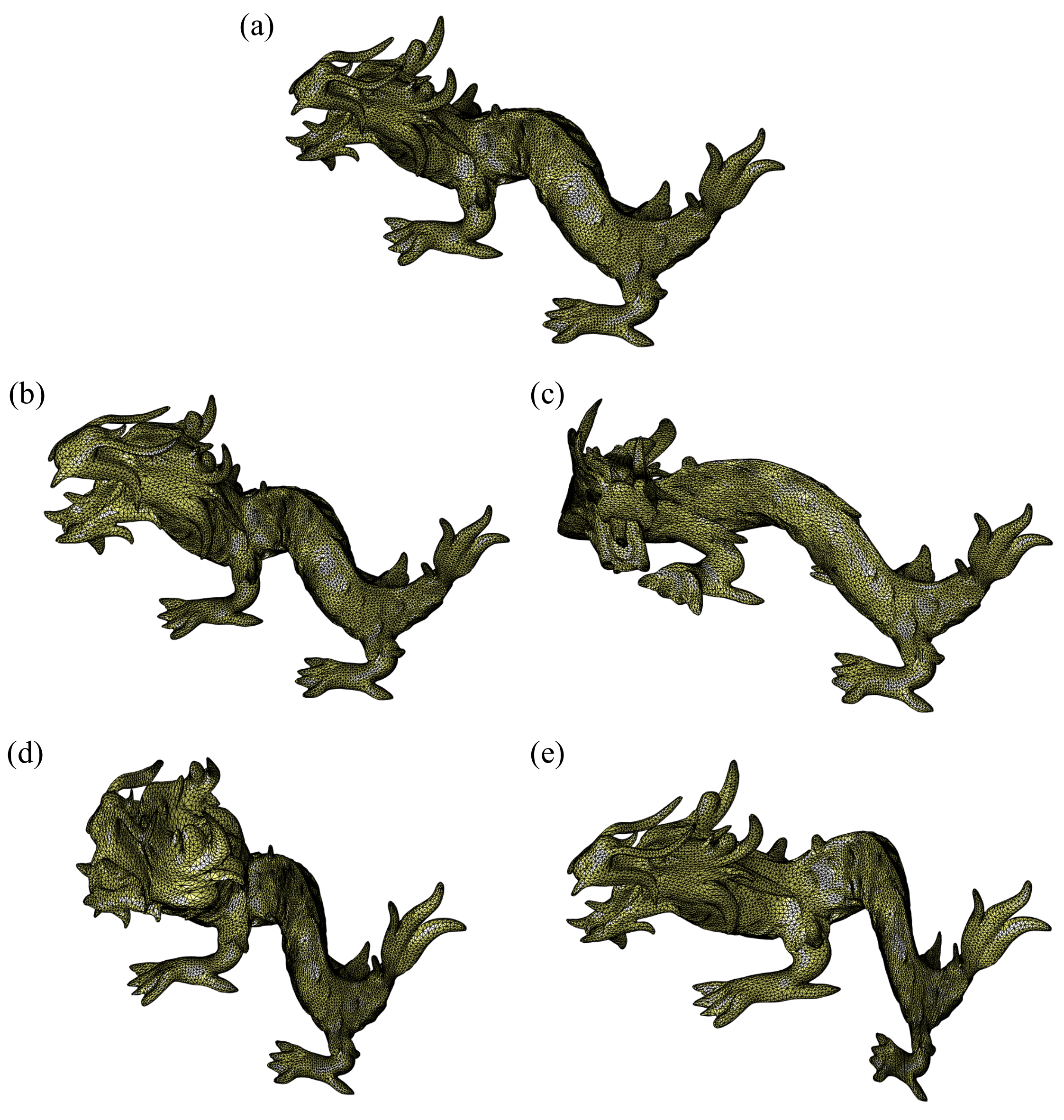}
\caption{Deformation-based shape modeling using the proposed framework. (a) The original dragon model from~\cite{stanford}. (b) A deformed dragon model with the head enlarged. (c) A deformed dragon model with the head rotated. (d) A deformed dragon model with the head lifted and the head and tail enlarged. (e) A deformed dragon model with the body twisted and shrunk.}\label{fig:dragon}
\end{figure}

To illustrate this idea, we consider a 3D dragon model adapted from The Stanford 3D Scanning Repository~\cite{stanford} (Fig.~\ref{fig:dragon}(a)). Note that the volumetric density-equalizing reference map method in~\cite{choi2021volumetric} is capable of producing a magnification of the dragon head by setting a higher density at a specific region of the volumetric domain. Analogously, by using the proposed model~\eqref{case3} and setting a large volume prior at the head of the dragon, we can precisely enlarge the head of the dragon as shown in Fig.~\ref{fig:dragon}(b). Moreover, using the proposed model~\eqref{case1} and setting the corners of the bounding box of the head as rotating landmarks, we can effectively turn the dragon head as shown in Fig.~\ref{fig:dragon}(c). It is also possible to change the position of the head and enlarge various parts of the dragon simultaneously by using the most general formulation of the proposed model~\eqref{PM} (see Fig.~\ref{fig:dragon}(d)). More specifically, landmarks are used for controlling the translational effect of the head, and volume priors are used for enlarging the head and the tail of the dragon. Similarly, one can twist and shrink the body of the dragon simultaneously (see Fig.~\ref{fig:dragon}(e)). The above results demonstrate the effectiveness of the proposed framework for producing a large variety of graphical effects.

%%%%%%%%%%%%%%%%%%%%%%%%%%%%%%%%%%%%%%%%%%%%%%%%%%%%%%%%%%%%%
\section{Conclusion and future works}\label{sect:conclusion}
In this work, we have proposed a unifying framework for computing $n$-dimensional quasi-conformal mappings. Specifically, our framework allows for the consideration of prescribed landmark and intensity information, volume prior, and the overall quasi-conformal and volumetric distortion. By adjusting the weights of different terms in the proposed energy model, we can easily achieve a large variety of mappings with different effects. 

For future work, we plan to employ the proposed framework on other real medical datasets for disease diagnosis. Another possible future direction is to explore the possibility of combining the framework with machine learning approaches~\cite{law2020cnn} for further improving $n$-dimensional image registrations. 

Also, note that the volumetric prior considered in the proposed framework is naturally related to optimal transportation (OT) maps~\cite{ur20093d,levy2015numerical,gu2016variational,su2017volume,lin20213d}. In particular, conventional optimal transportation maps are guaranteed to satisfy any prescribed volumetric constraints, and if we do not care about the optimality, there are infinitely many solutions according to Brenier's polar decomposition theorem~\cite{brenier1991polar}. When compared to the OT maps, our method considers additional components such as the generalized conformality distortion and the regularization and hence is likely to produce a smoother mapping which may not be exactly a solution of the discrete OT problem. Nevertheless, our method involves relatively simpler computation procedures and hence is likely to be more efficient. As a preliminary experiment, we consider a simple 2D problem of mapping a given initial Voronoi diagram to a new Voronoi diagram with prescribed area and compare the performance of our approach and the OT approach in~\cite{gu2016variational,omtcode}. For a domain of size $64\times 64$, our method only takes about 1.5 seconds while the OT approach in~\cite{gu2016variational,omtcode} takes about 4.7 seconds. For a domain of size $128\times 128$, our method takes about 4.5 seconds while the OT approach takes over 30 seconds. This suggests that our method is likely to be much more efficient than the existing OT approaches. Therefore, our framework with the volumetric prior may serve as a good initial map for the computation of optimal transportation maps.

Another natural next step is to extend the proposed framework for more general $n$-dimensional triangulations. Note that the current framework assumes the problem domain to be regular grids, as images often have a rectangular shape and the computation involves a differentiable interpolation operator for the deformed template (such as the cubic spline interpolation). If we only consider volume prior, landmark constraints and conformal distortion in searching for an optimal transformation, i.e. by using the proposed model~\eqref{PM} without the intensity fitting term, the finite element method is a good candidate to solve the problem for more general triangulations. In particular, the key step for the finite element implementation is to compute the Jacobian determinant of the transformation, which has already been studied in prior hyperelastic regularization image registration works~\cite{droske2004variational,ruthotto2012hyperelastic}. We plan to investigate the extension in detail in our future work.

\bibliographystyle{siam}
\bibliography{ndqc_bib.bib}

\end{document}